\documentclass{article}

\usepackage{PRIMEarxiv}
\usepackage[utf8]{inputenc}
\usepackage[T1]{fontenc}
\usepackage{hyperref}
\usepackage{url}
\usepackage{booktabs}
\usepackage{setspace}
\usepackage{amsmath,amssymb,amsthm,mathtools}
\usepackage{nicefrac}
\usepackage{microtype}
\usepackage{enumitem}
\usepackage{xcolor}
\usepackage{graphicx}
\usepackage{tikz}

\usetikzlibrary{arrows.meta,fit,positioning}

\numberwithin{equation}{section}

\newtheorem{theorem}{Theorem}[section]

\newtheorem{proposition}[theorem]{Proposition}
\newtheorem{corollary}[theorem]{Corollary}

\theoremstyle{definition}
\newtheorem{definition}[theorem]{Definition}
\newtheorem{assumption}[theorem]{Assumption}

\DeclareMathOperator{\argmin}{arg\,min}

\title{A Variational Framework for LLM Generator--Regulator Games}
\author{Quanyan Zhu\thanks{Department of Electrical and Computer Engineering, Tandon School of Engineering, New York University, Brooklyn, NY, USA. Contact: \href{mailto:quanyan.zhu@nyu.edu}{\texttt{quanyan.zhu@nyu.edu}}.}}

\begin{document}
\maketitle

\begin{abstract}
This paper develops a variational framework for regulated language generation. Starting from autoregressive token sampling, we derive the induced distribution over complete messages and relate it to an entropy-regularized Gibbs law. Regulation is modeled as an optimal discriminator whose convex-dual value is an $f$-divergence, and the generator--regulator interaction is formulated as a saddle-point problem. The framework applies to moderation, censorship, AI deception detection, compliance auditing, phishing defense, and manipulation control, where regulation concerns a distribution over possible messages rather than a single output. The equilibrium clarifies the tradeoff among utility, entropy, regulatory alignment, and finite-length detectability. Two finite-vocabulary case studies, censorship filtering and phishing defense, illustrate how the theory can be evaluated through utility, entropy, divergence, receiver-side scores, and detection probability.
\end{abstract}


\section{Introduction}

Large language models generate text, code, and structured reasoning one token at a time, yet their outputs are usually judged as complete messages. Semantic quality, safety, factuality, and compliance are properties of the realized sequence, not only of the next-token distribution. This gap motivates a message-level theory: starting from temperature-scaled autoregressive sampling, we study the induced law on complete messages and place it in a variational and game-theoretic framework. The construction draws on transformer language modeling \cite{vaswani2017attention,brown2020language}, information theory, entropy maximization, convex duality, and minimax equilibrium \cite{shannon1948mathematical,cover2006elements,rockafellar1970convex,sion1958general}.

The central modeling problem is that useful generation and regulated generation need not coincide. A generator should place mass on relevant, coherent, and useful messages, but it may also need to avoid output distributions that are unsafe, noncompliant, deceptive, censored, phishing-like, or otherwise undesirable. We use ``regulation'' broadly: the regulator may be a safety moderator, censorship mechanism, deception detector, phishing filter, compliance auditor, institutional policy, or adversarial monitor. This framing is closely related to game-theoretic models of cyber deception, defensive deception, adversarial evidence, and cyber resilience \cite{pawlick2017taxonomy,pawlick2018leaky,zhang2018deceptionfoundation,zhu2019cyberdeceptiontutorial,zhu2024cyberresilience}. In each case, regulation is a constraint on a message distribution rather than only a post-processing rule for one sample. This is crucial when policy-relevant behavior can be expressed through many token sequences or depends on semantic features such as deception, evasion, persuasion, credential-seeking intent, or forbidden topics. The model captures this tension by letting the generator choose a message law and the regulator choose a score that detects deviations from a regulated reference law.

The analysis proceeds from token sampling to message laws, from entropy-regularized utility maximization to a Gibbs generator, and from variational discrimination to divergence-regularized equilibrium. This sequence makes explicit how utility, diversity, and regulatory pressure jointly determine the generated distribution. A useful way to read the paper is to keep the following dictionary in mind: $Q$ is the law actually induced by the generator, $U_L$ measures the value of a complete length-$L$ message, $H(Q)$ rewards diversity, $Q_{\mathrm{reg}}$ is the reference law regarded as acceptable by the regulator, $\phi$ is a learned or designed score, and $\lambda$ determines how strongly the generator must respect the regulator.

Figure~\ref{fig:generator-regulator-model} summarizes the model. A prompt and context induce token-level sampling probabilities and hence a message law $Q$. The generator values utility and entropy; the regulator compares $Q$ with a reference law $Q_{\mathrm{reg}}$ by optimizing a discriminator $\phi$. The induced penalty $\mathcal R(Q)$ feeds back into the generator objective, and the saddle point determines the regulated equilibrium law.

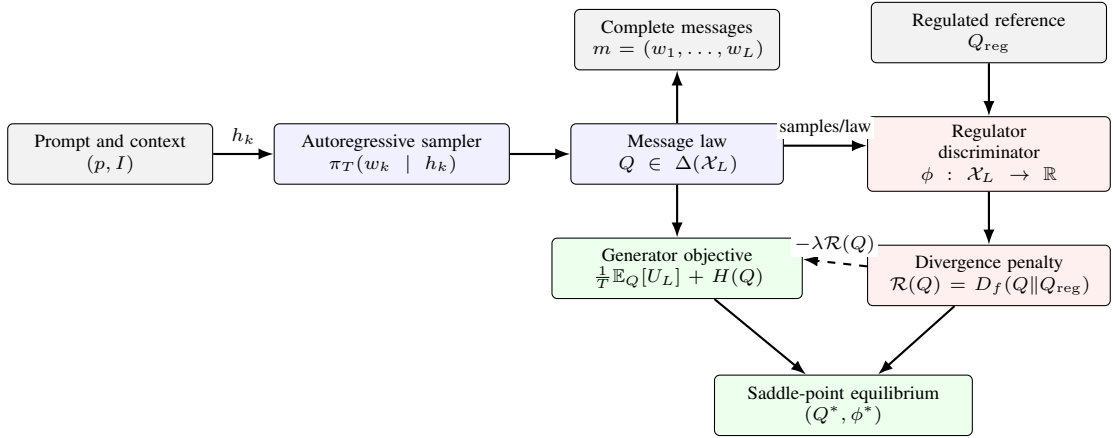
\begin{figure}[t]
\centering
\begin{tikzpicture}[
font=\scriptsize,
node distance=7mm and 8mm,
box/.style={draw, rounded corners=2pt, align=center, minimum height=8mm, inner sep=3pt, text width=25mm},
source/.style={box, fill=gray!10},
gen/.style={box, fill=blue!6},
reg/.style={box, fill=red!6},
value/.style={box, fill=green!7},
arrow/.style={-{Latex[length=2mm]}, thick},
feedback/.style={-{Latex[length=2mm]}, thick, dashed}
]
\node[source] (prompt) {Prompt and context\\$(p,I)$};
\node[gen, right=of prompt, text width=29mm] (sampler) {Autoregressive sampler\\$\pi_T(w_k\mid h_k)$};
\node[gen, right=of sampler, text width=26mm] (law) {Message law\\$Q\in\Delta(\mathcal X_L)$};
\node[source, above=of law, text width=25mm] (messages) {Complete messages\\$m=(w_1,\ldots,w_L)$};
\node[value, below=of law, text width=31mm] (objective) {Generator objective\\$\frac{1}{T}\mathbb E_Q[U_L]+H(Q)$};
\node[reg, right=11mm of law, text width=30mm] (disc) {Regulator\\discriminator\\$\phi:\mathcal X_L\to\mathbb R$};
\node[source, above=of disc, text width=29mm] (reference) {Regulated reference\\$Q_{\mathrm{reg}}$};
\node[reg, below=of disc, text width=30mm] (penalty) {Divergence penalty\\$\mathcal R(Q)=D_f(Q\|Q_{\mathrm{reg}})$};
\node[value, below=10mm of objective, xshift=22mm, text width=32mm] (equilibrium) {Saddle-point equilibrium\\$(Q^*,\phi^*)$};

\draw[arrow] (prompt) -- node[above] {$h_k$} (sampler);
\draw[arrow] (sampler) -- (law);
\draw[arrow] (law) -- (messages);
\draw[arrow] (law) -- (objective);
\draw[arrow] ([yshift=1.1mm]law.east) -- node[midway, above=2.2pt, fill=white, inner sep=1pt] {samples/law} ([yshift=1.1mm]disc.west);
\draw[arrow] (reference) -- (disc);
\draw[arrow] (disc) -- (penalty);
\draw[feedback] ([yshift=1.2mm]penalty.west) -- node[midway, above=2pt, fill=white, inner sep=1pt] {$-\lambda\mathcal R(Q)$} ([yshift=1.2mm]objective.east);
\draw[arrow] (objective) -- (equilibrium);
\draw[arrow] (penalty) -- (equilibrium);
\end{tikzpicture}
\caption{Schematic of the variational generator--regulator model. Autoregressive token sampling induces a message distribution $Q$ on the fixed-length message space $\mathcal X_L$. The generator rewards utility and entropy, while the regulator compares $Q$ with a reference distribution encoding safety, censorship, nondeception, phishing defense, compliance, or another regulated behavior. The optimized discriminator induces a divergence penalty, closing the feedback loop and determining the saddle-point equilibrium.}
\label{fig:generator-regulator-model}
\end{figure}

\subsection{Contributions and Novelty}

The paper makes six contributions: a message-level probabilistic formulation of autoregressive generation; a Gibbs representation derived from entropy-regularized utility maximization; a variational discriminator model of regulation with an $f$-divergence dual; a saddle-point characterization with a closed-form KL-regulated equilibrium; finite-length Gaussian approximations for score, log-density, and surprisal statistics; and two finite-vocabulary case studies, censorship filtering and phishing defense, that evaluate utility, entropy, regulatory divergence, receiver-side scoring, and detection probability.

The novelty is the integration of generation, moderation, censorship, deception detection, phishing defense, and compliance into one distributional variational structure. The generator is a message-level Gibbs law, the regulator is an optimal discriminator whose value is an $f$-divergence, and their interaction is a saddle-point problem. This yields an explicit KL-regulated equilibrium and connects finite-length detectability to central limit approximations for score, log-density, and surprisal.

\section{Generator Model and Variational Characterization}

This section connects the operational mechanism of autoregressive language models with a global variational formulation of generation. The central observation is that token-level sampling induces a probability distribution over complete messages, and this induced distribution admits an energy--entropy interpretation.

\subsection{Autoregressive Generation in Practice}

Modern large language models generate text through a sequential stochastic process known as \emph{autoregressive generation}: a message is constructed token by token, with each new token conditioned on the previously generated prefix and the external context. Let $\mathcal V$ denote a finite vocabulary, where each element $w \in \mathcal V$ is a token representing a word, subword unit, or symbol depending on the tokenization scheme. Let $p$ denote the prompt, and let $I$ denote auxiliary contextual information such as system-level instructions, retrieved documents, memory states, or external signals. Together, $(p,I)$ define the conditioning information under which generation takes place.

At generation step $k$, the model conditions on the history $h_k=(w_1,\dots,w_{k-1},p,I)$, which aggregates all information available before the $k$-th token is generated. This history plays the role of a state variable in a sequential stochastic process, and the autoregressive assumption implies that the next-token distribution depends on $h_k$ rather than on future tokens. Given the history $h_k$, the model computes real-valued logits $\ell(w \mid h_k)$ for $w \in \mathcal V$. These logits are unnormalized preferences over the vocabulary; they arise from the internal neural network computation and encode the model's assessment of each token as the next continuation. Since logits live in $\mathbb R$ rather than in the probability simplex, the model converts them into probabilities using the temperature-scaled softmax transformation
\begin{equation}
\pi_T(w \mid h_k)
=
\frac{\exp\big(\ell(w \mid h_k)/T\big)}
{\sum_{v \in \mathcal V} \exp\big(\ell(v \mid h_k)/T\big)}.
\label{eq:softmax}
\end{equation}
This mapping produces a valid probability distribution over $\mathcal V$, with $\sum_w \pi_T(w \mid h_k)=1$, and it is invariant to additive shifts of the logits, so only relative differences matter. The temperature parameter $T$ controls the sharpness of the distribution: smaller values amplify logit differences and yield more deterministic behavior, while larger values flatten the distribution and increase randomness.

Once the conditional distribution $\pi_T(\cdot \mid h_k)$ is defined, the model samples the next token $w_k$ according to this distribution and repeats the procedure until a full message $m=(w_1,\dots,w_L)$ is produced. The probability of generating a specific message is therefore the product of conditional probabilities,
\begin{equation}
P_T(m \mid p, I)
=
\prod_{k=1}^L \pi_T(w_k \mid h_k).
\label{eq:autoregressive-product}
\end{equation}
This expression follows directly from the chain rule of probability and gives the operational definition of generation used in practice. The sequence $\{w_k\}_{k=1}^L$ can be viewed as a stochastic process adapted to the filtration generated by the histories $\{h_k\}$: each step introduces randomness through sampling, while dependence on the history enforces coherence and context-awareness.

\subsection{Message Space and Induced Distribution}

To move from token-level generation to a global perspective, we introduce the message space. For a fixed length $L$, define $\mathcal X_L:=\mathcal V^L$, and let the variable-length message space be $\mathcal X=\bigcup_{L\ge 1}\mathcal X_L$. This construction allows sequences of arbitrary finite length; each element $m \in \mathcal X$ can be written as $m=(w_1,\dots,w_L)$ for some $L\ge 1$, where each $w_k\in\mathcal V$. In all finite-length results below, however, the relevant space is the finite set $\mathcal X_L$. The autoregressive generation mechanism in \eqref{eq:autoregressive-product} assigns a probability to every sequence $m \in \mathcal X_L$. For fixed $L$, this expression induces a probability measure $P_T(\cdot\mid p,I)\in\Delta(\mathcal X_L)$, where $\Delta(\mathcal X_L)$ denotes the probability simplex over $\mathcal X_L$. The local normalization $\sum_{w \in \mathcal V} \pi_T(w \mid h_k)=1$ ensures that the global distribution over sequences is properly normalized. Under an appropriate stopping rule, such as an end-of-sequence token, the chain rule gives $\sum_{m \in \mathcal X} P_T(m \mid p, I)=1$, so local normalization at every token position is enough to define a global probability law.

This message-level law is the central object of the paper. A next-token distribution tells us what may happen at one step, but questions about safety, censorship, phishing risk, deception, or factuality are usually questions about a complete message. By passing from $\pi_T(w_k\mid h_k)$ to $Q$ on $\mathcal X_L$, we can ask how much probability mass the model assigns to whole classes of messages, not only whether a single sampled output is acceptable.
 
\subsection{From Sequential Normalization to Global Gibbs Form}

The path-dependent normalization in the autoregressive product complicates sequence-level analysis, because the normalizing factor at position $k$ depends on the realized prefix. To expose the energy--entropy structure of generation, we consider a message-level Gibbs approximation that replaces the sequence of local normalization terms with a single global partition function. Specifically, for fixed $L$ we write
\begin{equation}
Q_T(m)
\approx
\frac{\exp\left(\frac{1}{T} U_L(m)\right)}
{Z_L(T)},
\label{eq:gibbs}
\end{equation}
where the global partition function is defined as
$Z_L(T)=\sum_{m \in \mathcal V^L}\exp\left(U_L(m)/T\right)$.
Under this approximation, the distribution over sequences takes the form of a Gibbs distribution: the trajectory-dependent normalization is replaced by a single scalar quantity that aggregates contributions from all sequences of length $L$. This approximation is most useful in regimes where aggregate behavior dominates, such as long-sequence limits or studies of typical generated messages, because fluctuations from individual normalization terms can often be absorbed into a global normalization. The quantity $U_L(m)$ is a sequence-level utility that scores semantic quality, relevance, coherence, or task value; the Gibbs form assigns exponentially larger probability mass to higher-utility messages, while $Z_L(T)$ normalizes this preference across all messages. Low temperature concentrates mass on high-utility messages, whereas high temperature gives relatively greater weight to entropy and diversity. This is the same mathematical structure that underlies maximum-entropy models and log-partition variational principles \cite{cover2006elements,rockafellar1970convex,donsker1975asymptotic}.

\subsection{Variational Interpretation}

We next identify the variational objective behind \eqref{eq:gibbs}: expected utility rewards useful messages, while entropy prevents collapse onto a small set of high-scoring sequences.

\begin{definition}[Entropy]
Let $Q \in \Delta(\mathcal X)$ be a probability distribution over the message space $\mathcal X$. The Shannon entropy of $Q$ is $H(Q):=-\sum_{m \in \mathcal X} Q(m)\log Q(m)$.
\end{definition}

\begin{definition}[Variational Objective]
For a fixed temperature $T>0$, define the functional $\mathcal J_L : \Delta(\mathcal X) \to \mathbb R$ by
\begin{equation}
\mathcal J_L(Q)
=
\frac{1}{T} \mathbb{E}_Q[U_L(m)]
+
H(Q).
\label{eq:variational-functional}
\end{equation}
\end{definition}

\begin{proposition}[Variational Characterization of the Gibbs Distribution]
\label{prop:variational}
The optimization problem $\sup_{Q \in \Delta(\mathcal X)} \mathcal J_L(Q)$
admits a unique solution given by the Gibbs distribution
\begin{equation}
Q^*(m)
=
\frac{\exp\left(\frac{1}{T} U_L(m)\right)}
{Z_L(T)},
\end{equation}
where the partition function is
$Z_L(T)=\sum_{m \in \mathcal V^L}\exp(U_L(m)/T)$.
\end{proposition}

\begin{proof}
The result follows from standard convex duality arguments. Introducing a Lagrange multiplier $c$ for the normalization constraint $\sum_m Q(m)=1$, the first-order optimality condition yields
$\log Q^*(m)=U_L(m)/T-c$, which implies the stated form after normalization.
\end{proof}

The functional $\mathcal J(Q)$ balances expected utility and entropy; $T$ controls concentration on high-utility sequences versus dispersion across the message space. For students, the proposition can be read as a maximum-entropy principle: among all laws over messages, the Gibbs law is the unique law that gives high probability to useful messages while remaining as spread out as possible subject to that preference. The partition function $Z_L(T)$ is not merely a normalizing constant; its logarithm is the optimal value of the entropy-regularized generation problem.

\begin{definition}[Normalized Utility]
Let $m = (w_1,\dots,w_L) \in \mathcal X$ be a message of length $L$. The normalized utility is defined as $\bar U_L(m):=U_L(m)/L$.
\end{definition}

The normalized utility $\bar U_L(m)$ represents the average contribution per token. This normalization is essential when comparing sequences of different lengths and plays a central role in asymptotic analysis as $L \to \infty$.

\begin{proposition}[Optimal Value of the Variational Problem]
Let $Q^*$ denote the optimizer in Proposition~\ref{prop:variational}. Then the optimal value is $\sup_{Q \in \Delta(\mathcal X)} \mathcal J_L(Q)=\log Z_L(T)$. Consequently, the normalized optimal value per token is $\frac{1}{L}\log Z_L(T)$.
\end{proposition}

\begin{proof}
For the Gibbs distribution $Q^*(m)=\exp(U_L(m)/T)/Z_L(T)$, we have $\log Q^*(m)=U_L(m)/T-\log Z_L(T)$. Hence
\[
\mathcal J_L(Q^*)
=
\mathbb E_{Q^*}\!\left[\frac{1}{T}U_L(m)-\log Q^*(m)\right]
=
\log Z_L(T).
\]
Proposition~\ref{prop:variational} shows that $Q^*$ is the unique maximizer, so this value is the supremum. Dividing by $L$ gives the normalized value.
\end{proof}

\begin{definition}[Generative Free Energy Density]
Assuming the limit exists, define the generative free energy density at temperature $T>0$ as
\begin{equation}
\mathcal F(T)
:=
\lim_{L \to \infty}
\frac{1}{L} \log Z_L(T),
\end{equation}
where the partition function is given by
$Z_L(T)=\sum_{m \in \mathcal V^L}\exp(U_L(m)/T)$.
\end{definition}

The quantity $\frac{1}{L}\log Z_L(T)$ is the optimal entropy-regularized utility per token. When the limit exists, $\mathcal F(T)$ gives the exponential growth rate $Z_L(T)\asymp \exp(L\mathcal F(T))$ of the utility-weighted message count and measures the effective richness of high-utility sequences. Larger values indicate a richer set of useful messages; smaller values indicate scarcity or a sharper utility--diversity tradeoff.

\section{Regulator Model}

We introduce a regulator as a distributional object. It may represent a moderator, censor, safety classifier, deception detector, phishing filter, compliance auditor, or policy-enforcing mechanism, but it constrains the law of generated messages rather than only a realized sample. Thus a generator can respond to regulation by changing the probability of whole message classes, such as forbidden topics, misleading or evasive messages, or credential-seeking risk features. Let $\mathcal X$ be the message space and $\Delta(\mathcal X)$ the set of distributions on it. A regulator is a functional $\mathcal R:\Delta(\mathcal X)\to\mathbb R$, where $\mathcal R(Q)$ measures how far the generator law $Q$ deviates from prescribed regulated behavior.

\subsection{General Variational Form}

We model regulation as optimal discrimination, following variational views of adversarial discriminators and $f$-divergence estimation \cite{goodfellow2014generative,nowozin2016fgan,csiszar1967information}. Let $\Phi$ be a class of measurable scores $\phi:\mathcal X\to\mathbb R$. In applications, $\phi$ may score prohibited topics, inconsistency, hidden intent, credential requests, urgency cues, or link-risk features. Let $\psi:\mathbb R\to\mathbb R$ be proper, convex, and lower semicontinuous, so that it regularizes large or irregular scores. Given a regulated reference distribution $Q_{\mathrm{reg}}\in\Delta(\mathcal X)$, define
\begin{equation}
\mathcal R(Q)
=
\sup_{\phi \in \Phi}
\left\{
\mathbb E_Q[\phi(m)]
-
\mathbb E_{Q_{\mathrm{reg}}}[\psi(\phi(m))]
\right\}.
\label{eq:regulator-general}
\end{equation}
For fixed $Q$, the regulator selects the score that best separates samples from $Q$ and samples from $Q_{\mathrm{reg}}$. If the two laws are close, no score achieves much separation and $\mathcal R(Q)$ is small; if they differ, the optimal discriminator has large value. Convex duality turns this testing problem into a divergence penalty, linking operational detection to analytic optimization.

\subsubsection{Convex Duality and Divergence Representation}

Convex duality gives the link between the variational regulator and divergence measures. The best test for separating two probability laws is governed by their likelihood ratio, and convex losses convert that separation problem into a divergence. Thus, with a sufficiently rich score class, the regulator induces an $f$-divergence penalty between the generator distribution and the regulated reference law \cite{csiszar1967information,rockafellar1970convex,nowozin2016fgan}.

\begin{definition}[Convex Conjugate]
Let $\psi : \mathbb R \to \mathbb R \cup \{+\infty\}$ be a proper, convex, and lower semicontinuous function. The convex conjugate $\psi^* : \mathbb R \to (-\infty,+\infty]$ is defined by $\psi^*(u)=\sup_{t\in\mathbb R}\{ut-\psi(t)\}$.
\end{definition}

The convex conjugate $\psi^*$ provides the dual representation used below. The main representation result is the following.

\begin{theorem}[Divergence Representation of the Regulator]
\label{thm:regulator-divergence}
Let $\Phi$ be a class of bounded measurable functions that is rich enough to approximate any bounded measurable function on $\mathcal X$. Let $\psi$ be proper, convex, and lower semicontinuous. 

Then, for any $Q \in \Delta(\mathcal X)$ such that $Q \ll Q_{\mathrm{reg}}$, the regulator admits the representation
\begin{equation}
\mathcal R(Q)
=
\mathbb E_{Q_{\mathrm{reg}}}
\left[
\psi^*\!\left(\frac{dQ}{dQ_{\mathrm{reg}}}(m)\right)
\right].
\label{eq:f-divergence}
\end{equation}
\end{theorem}

\begin{proof}
Let $Q \ll Q_{\mathrm{reg}}$ and define the Radon--Nikodym derivative $r(m):=dQ/dQ_{\mathrm{reg}}(m)$.
For any $\phi \in \Phi$, we rewrite the first term in \eqref{eq:regulator-general} as
\begin{equation}
\mathbb E_Q[\phi]
=
\mathbb E_{Q_{\mathrm{reg}}}
\big[ \phi(m) r(m) \big].
\end{equation}
Substituting into \eqref{eq:regulator-general}, we obtain
\begin{equation}
\mathcal R(Q)
=
\sup_{\phi \in \Phi}
\mathbb E_{Q_{\mathrm{reg}}}
\Big[
\phi(m) r(m) - \psi(\phi(m))
\Big].
\end{equation}

Since $\Phi$ is sufficiently rich, the supremum over $\phi$ can be taken pointwise. That is, for each $m \in \mathcal X$, we optimize independently:
\begin{equation}
\sup_{t \in \mathbb R}
\big\{ t\, r(m) - \psi(t) \big\}
=
\psi^*\big(r(m)\big).
\end{equation}

Therefore,
\begin{equation}
\mathcal R(Q)
=
\mathbb E_{Q_{\mathrm{reg}}}
\left[
\psi^*\big(r(m)\big)
\right]
=
\mathbb E_{Q_{\mathrm{reg}}}
\left[
\psi^*\!\left(\frac{dQ}{dQ_{\mathrm{reg}}}(m)\right)
\right],
\end{equation}
which establishes \eqref{eq:f-divergence}.
\end{proof}

The theorem shows that the regulator is an $f$-divergence with generator function $f = \psi^*$. Intuitively, the likelihood ratio $dQ/dQ_{\mathrm{reg}}$ identifies where the generator puts more mass than the regulated reference. The convex function $\psi^*$ then decides how severely to penalize these relative-density deviations. Different choices of $\psi$ therefore correspond to different regulatory attitudes toward mismatch: some penalize moderate deviations gently, while others punish concentrated deviations sharply.

\begin{corollary}[Optimal Discriminator]
\label{cor:optimal-discriminator}
Under the assumptions of Theorem~\ref{thm:regulator-divergence}, suppose further that $\psi$ is differentiable and strictly convex. Then the supremum in \eqref{eq:regulator-general} is attained by a function $\phi^* \in \Phi$ satisfying
\begin{equation}
\phi^*(m)
=
\left(\psi'\right)^{-1}
\!\left(
\frac{dQ}{dQ_{\mathrm{reg}}}(m)
\right)
\quad
\text{for $Q_{\mathrm{reg}}$-almost every } m.
\end{equation}
\end{corollary}

\begin{proof}
From the proof of Theorem~\ref{thm:regulator-divergence}, the optimization over $\phi$ reduces pointwise to
\begin{equation}
\sup_{t \in \mathbb R}
\{ t\, r(m) - \psi(t) \},
\end{equation}
where $r(m) = \frac{dQ}{dQ_{\mathrm{reg}}}(m)$.

Since $\psi$ is differentiable and strictly convex, the supremum is achieved at the unique maximizer $t^*(m)$ satisfying the first-order optimality condition
\begin{equation}
r(m) - \psi'(t^*(m)) = 0.
\end{equation}
Solving for $t^*(m)$ yields
\begin{equation}
t^*(m)
=
\left(\psi'\right)^{-1}(r(m)).
\end{equation}
Setting $\phi^*(m) = t^*(m)$ gives the desired result.
\end{proof}

The optimal discriminator $\phi^*$ depends explicitly on the likelihood ratio $\frac{dQ}{dQ_{\mathrm{reg}}}$ and therefore extracts precisely the evidence that distinguishes the generated distribution from the regulated reference. This provides a clear interpretation of the regulator as an optimal statistical test between two message laws, rather than as an ad hoc rule applied after generation.

 \subsubsection{Optimal Regulator}

The previous corollary identifies the optimal regulator explicitly: under differentiability and strict convexity, the best scoring function is a monotone transform of the likelihood ratio $dQ/dQ_{\mathrm{reg}}$. Thus the regulator depends only on the statistical evidence that separates the generator distribution from the regulated reference.

\begin{definition}[$f$-Divergence]
Let $f : \mathbb R_+ \to \mathbb R$ be a convex function with $f(1)=0$. The $f$-divergence between $Q$ and $Q_{\mathrm{reg}}$ is defined as
\begin{equation}
D_f(Q \,\|\, Q_{\mathrm{reg}})
=
\mathbb E_{Q_{\mathrm{reg}}}
\left[
f\!\left(\frac{dQ}{dQ_{\mathrm{reg}}}(m)\right)
\right].
\end{equation}
\end{definition}

\begin{proposition}[Regulator as a Divergence]
Under the assumptions of Theorem~\ref{thm:regulator-divergence}, let $f(u):=\psi^*(u)-\psi^*(1)$. Then
\begin{equation}
\mathcal R(Q)
=
D_f(Q\,\|\,Q_{\mathrm{reg}})
+
\psi^*(1).
\end{equation}
In particular, when $\psi^*(1)=0$, the regulator is exactly the $f$-divergence generated by $f=\psi^*$.
\end{proposition}

\begin{proof}
By Theorem~\ref{thm:regulator-divergence},
\[
\mathcal R(Q)
=
\mathbb E_{Q_{\mathrm{reg}}}
\left[
\psi^*\!\left(\frac{dQ}{dQ_{\mathrm{reg}}}(m)\right)
\right].
\]
Writing $\psi^*(u)=f(u)+\psi^*(1)$ gives
\[
\mathcal R(Q)
=
\mathbb E_{Q_{\mathrm{reg}}}
\left[
f\!\left(\frac{dQ}{dQ_{\mathrm{reg}}}(m)\right)
\right]
+
\psi^*(1)
=
D_f(Q\,\|\,Q_{\mathrm{reg}})+\psi^*(1).
\]
Since $f(1)=0$, $f$ is a normalized $f$-divergence generator.
\end{proof}

The representation \eqref{eq:f-divergence} shows that the regulator measures the discrepancy between $Q$ and $Q_{\mathrm{reg}}$ through a convex function of the likelihood ratio $\frac{dQ}{dQ_{\mathrm{reg}}}$. It therefore penalizes deviations from regulated behavior through the relative density of the two distributions. Operationally, this means that a message region is risky for the generator when it is common under $Q$ but uncommon under $Q_{\mathrm{reg}}$; such a region is precisely where an optimal discriminator has evidence that the generator has departed from the regulated reference.

\subsubsection{Examples of Divergence Regulators}

The regulator functional $\mathcal R(Q)$ admits explicit and tractable forms for several important choices of the convex function $\psi$. These cases are of particular interest because they lead to well-known divergence measures and admit closed-form expressions for both the regulator and the optimal discriminator.

\emph{Kullback--Leibler divergence.}
Consider the choice $\psi(t)=e^t$. Its convex conjugate is $\psi^*(u)=u\log u-u+1$, which is the generator of the Kullback--Leibler (KL) divergence, a central discrepancy measure in information theory and statistical inference \cite{cover2006elements}. Substituting into \eqref{eq:f-divergence}, we obtain
\begin{equation}
\mathcal R(Q)
=
\mathbb E_{Q_{\mathrm{reg}}}
\left[
\frac{dQ}{dQ_{\mathrm{reg}}}(m)
\log \frac{dQ}{dQ_{\mathrm{reg}}}(m)
\right]
=
D_{\mathrm{KL}}(Q \,\|\, Q_{\mathrm{reg}}).
\end{equation}

In this case, the optimal discriminator from Corollary~\ref{cor:optimal-discriminator} is $\phi^*(m)=\log(dQ/dQ_{\mathrm{reg}})(m)$. This has a clear statistical interpretation: the optimal discriminator assigns to each message the log-likelihood ratio between $Q$ and $Q_{\mathrm{reg}}$, the canonical statistic for distinguishing the two distributions. From a computational perspective, the KL case is particularly attractive because it is convex, admits efficient estimation in many settings, and leads to a modified Gibbs distribution that incorporates a log-likelihood penalty.

\emph{Pearson $\chi^2$ divergence.}
Consider the quadratic choice $\psi(t)=\frac{1}{2}t^2$. Its convex conjugate is $\psi^*(u)=\frac{1}{2}u^2$.
Substituting into \eqref{eq:f-divergence}, we obtain
\begin{equation}
\mathcal R(Q)
=
\frac{1}{2}
\mathbb E_{Q_{\mathrm{reg}}}
\left[
\left(\frac{dQ}{dQ_{\mathrm{reg}}}(m)\right)^2
\right].
\end{equation}

Up to normalization, this corresponds to the Pearson $\chi^2$ divergence:
\begin{equation}
\chi^2(Q \,\|\, Q_{\mathrm{reg}})
=
\mathbb E_{Q_{\mathrm{reg}}}
\left[
\left(\frac{dQ}{dQ_{\mathrm{reg}}}(m) - 1\right)^2
\right].
\end{equation}

The optimal discriminator in this case is the likelihood ratio itself. This regulator penalizes large deviations of $Q$ from $Q_{\mathrm{reg}}$ more aggressively than the KL divergence, as the penalty grows quadratically in the likelihood ratio. As a result, it is particularly sensitive to outliers or rare events where $Q$ places excessive mass relative to $Q_{\mathrm{reg}}$.

 \emph{Jensen--Shannon divergence and the GAN regulator.}
Consider the choice $\psi(t)=\log(1+e^t)$.
This function is convex and corresponds to the logistic loss commonly used in binary classification. Its convex conjugate is given by
\begin{equation}
\psi^*(u)
=
u \log u + (1-u)\log(1-u),
\qquad u \in (0,1),
\end{equation}
and $+\infty$ otherwise.

With the standard GAN reparameterization and equal class priors, this choice leads to the Jensen--Shannon (JS) divergence, up to additive constants:
\begin{equation}
\mathcal R(Q)
\;=\;
\mathrm{JS}(Q \,\|\, Q_{\mathrm{reg}}),
\end{equation}
which is a symmetrized and smoothed version of the Kullback--Leibler divergence.

In this case, it is convenient to parametrize the discriminator through the probability function $D(m)=e^{\phi(m)}/(1+e^{\phi(m)})$.
The optimal discriminator takes the well-known form
\begin{equation}
D^*(m)
=
\frac{Q(m)}{Q(m) + Q_{\mathrm{reg}}(m)},
\end{equation}
and equivalently $\phi^*(m)=\log(Q(m)/Q_{\mathrm{reg}}(m))$.

This case has an important interpretation. The regulator corresponds to the objective used in generative adversarial networks (GANs): the discriminator attempts to distinguish samples from $Q$ and $Q_{\mathrm{reg}}$, and the Jensen--Shannon divergence measures how well the two distributions can be separated by an optimal classifier \cite{goodfellow2014generative,nowozin2016fgan}. Compared to the KL divergence, the JS divergence is symmetric and bounded:
$0 \le \mathrm{JS}(Q \,\|\, Q_{\mathrm{reg}}) \le \log 2$. This boundedness makes it numerically stable in many applications, although it may become insensitive when the supports of $Q$ and $Q_{\mathrm{reg}}$ are nearly disjoint, which can lead to optimization difficulties.

\subsubsection{Feature-Based Scoring and Detection}

In practical systems, regulation often proceeds by first extracting features from a message and then applying a decision rule based on these features. This motivates a two-stage architecture consisting of feature extraction and scoring followed by detection. The abstraction is deliberately broad: features may be lexical, syntactic, semantic, retrieval-based, classifier-based, or policy-specific, and the theory only requires that they define measurable summaries of complete messages. Let $\mathcal F : \mathcal X \to \mathbb R^d$ denote a feature mapping that associates each message $m \in \mathcal X$ with a feature vector $z=\mathcal F(m)$. The feature map $\mathcal F$ may depend on external information sources, such as knowledge bases or retrieval systems. Two important instances are as follows.

\emph{(i) Keyword-based features.}  
Let $\mathcal K \subseteq \mathcal V$ be a set of regulated tokens. Define
\begin{equation}
\mathcal F_{\mathrm{key}}(m)
=
\frac{1}{L} \sum_{k=1}^L \mathbf{1}\{w_k \in \mathcal K\}.
\end{equation}
This yields a scalar feature representing the fraction of regulated tokens.

\emph{(ii) Knowledge-based (RAG) features.}   
Let $\mathcal D$ denote an external knowledge base, such as a document corpus, database, or knowledge graph. Let $\mathrm{Retrieve}:\mathcal X\times\mathcal D\to 2^{\mathcal D}$ be a retrieval operator that, given a message $m \in \mathcal X$, returns a finite subset $\mathcal D_m:=\mathrm{Retrieve}(m,\mathcal D)\subseteq\mathcal D$ of documents or knowledge items deemed relevant to $m$.

We define the knowledge-based feature map $\mathcal F_{\mathrm{rag}} : \mathcal X \to \mathbb R^d$ by
\begin{equation}
\mathcal F_{\mathrm{rag}}(m)
=
g(m, \mathcal D_m),
\label{eq:rag-feature}
\end{equation}
where $g : \mathcal X \times 2^{\mathcal D} \to \mathbb R^d$ is a measurable function that extracts features by comparing the message $m$ with the retrieved knowledge $\mathcal D_m$.

The function $g$ may encode various forms of semantic or policy-relevant evaluation. For instance, it may measure factual consistency by comparing statements in $m$ with retrieved evidence, policy compliance by detecting alignment or violation relative to known rules, or semantic similarity by computing embeddings and distances between $m$ and elements of $\mathcal D_m$. In this formulation, the feature map $\mathcal F_{\mathrm{rag}}(m)$ depends not only on the message itself but also on external information retrieved from $\mathcal D$, distinguishing knowledge-based regulation from purely content-based methods. In a common special case, the feature extractor aggregates pairwise scores between the message and retrieved documents. For example, if $\mathcal D_m=\{d_1,\dots,d_K\}$, one may define
\begin{equation}
g(m,\mathcal D_m)
=
\frac{1}{K} \sum_{i=1}^K s(m,d_i),
\end{equation}
where $s : \mathcal X \times \mathcal D \to \mathbb R^d$ is a similarity or consistency score. This makes explicit how retrieval augments the feature representation through external knowledge.

Given the feature representation $z = \mathcal F(m) \in \mathbb R^d$, the regulator assigns a scalar score to the message through a function $s : \mathbb R^d \to \mathbb R$, yielding $S(m)=s(\mathcal F(m))$. The scoring function aggregates multiple features into a single quantity that measures the degree to which the message exhibits regulated characteristics. It may be linear, nonlinear, or learned from data; for example, a linear rule takes the form $S(m)=\langle \theta,\mathcal F(m)\rangle$ for some parameter vector $\theta\in\mathbb R^d$, while nonlinear choices can capture feature interactions. Statistically, $S(m)$ is a low-dimensional summary statistic for the message, reducing the complexity of $\mathcal X$ to a one-dimensional quantity used for subsequent decision-making. Given the score $S(m)$, the regulator applies a decision rule to determine whether the message violates regulatory constraints. A canonical choice is a threshold-based rule
\begin{equation}
\delta(m)
=
\mathbf{1}\{ S(m) \ge \eta \},
\end{equation}
where $\eta \in \mathbb R$ is a threshold parameter.
This rule partitions the message space into accepted and flagged regions, and the threshold $\eta$ controls the detector's operating point by trading off false positives and false negatives.

We next present a formal equivalence between the feature-based detection problem and the variational regulator in \eqref{eq:regulator-general}. The result shows that the regulator arises as the dual representation of a convex risk minimization problem, thereby connecting the decision-theoretic interpretation of regulation with the divergence-based formulation above.

\begin{definition}[Binary Decision Problem]
Let $Q_{\mathrm{reg}}, Q \in \Delta(\mathcal X)$ be two probability measures. A decision rule is a measurable function $\delta : \mathcal X \to [0,1]$. The associated risk is defined as
\begin{equation}
\mathcal L(\delta)
=
\mathbb E_{Q_{\mathrm{reg}}}[\ell_0(\delta(m))]
+
\mathbb E_Q[\ell_1(\delta(m))],
\label{eq:risk-formal}
\end{equation}
where $\ell_0,\ell_1 : [0,1] \to \mathbb R$ are convex loss functions.
\end{definition}

\begin{assumption}
The losses $\ell_0$ and $\ell_1$ are convex, proper, and differentiable, and induce a strictly proper scoring rule.
\end{assumption}

\begin{proposition}[Score-Based Reformulation]
\label{prop:score-reformulation-clean}
Let $Q_{\mathrm{reg}}, Q \in \Delta(\mathcal X)$ and consider the risk \eqref{eq:risk-formal}
where $\ell_0,\ell_1 : [0,1] \to \mathbb R$ are convex functions. 
Assume that $\ell_0$ is strictly convex and differentiable. Then there exists a convex function $\psi : \mathbb R \to \mathbb R$ such that
\begin{equation}
\inf_{\delta} \mathcal L(\delta)
=
C
-
\sup_{\phi}
\left\{
\mathbb E_Q[\phi(m)]
-
\mathbb E_{Q_{\mathrm{reg}}}[\psi(\phi(m))]
\right\},
\end{equation}
for some constant $C$ independent of $\phi$.
\end{proposition}

\begin{proof}
Assume $Q \ll Q_{\mathrm{reg}}$ and define the likelihood ratio $r(m):=dQ/dQ_{\mathrm{reg}}(m)$.

Then the risk can be written as
\[
\mathcal L(\delta)
=
\mathbb E_{Q_{\mathrm{reg}}}
\left[
\ell_0(\delta(m)) + r(m)\,\ell_1(\delta(m))
\right].
\]

Since $\delta(m)$ is chosen pointwise, the minimization separates over $m$, yielding
\[
\inf_{\delta} \mathcal L(\delta)
=
\mathbb E_{Q_{\mathrm{reg}}}
\left[
\inf_{d \in [0,1]}
\left\{
\ell_0(d) + r(m)\,\ell_1(d)
\right\}
\right].
\]

Define $f(u):=\inf_{d \in [0,1]}\{\ell_0(d)+u\,\ell_1(d)\}$. Then $\inf_{\delta} \mathcal L(\delta)=\mathbb E_{Q_{\mathrm{reg}}}[f(r(m))]$.

The function $f$ is concave in $u$ as an infimum of affine functions. By concave duality, there exists a convex function $\psi$ such that
\[
f(u)
=
\inf_{t \in \mathbb R}
\left\{
\psi(t) - u t
\right\}.
\]

Substituting,
\[
\inf_{\delta} \mathcal L(\delta)
=
\mathbb E_{Q_{\mathrm{reg}}}
\left[
\inf_{t}
\left\{
\psi(t) - r(m)\,t
\right\}
\right].
\]

Under standard measurability conditions, we can interchange infimum and expectation to obtain
\[
\inf_{\delta} \mathcal L(\delta)
=
\inf_{\phi}
\mathbb E_{Q_{\mathrm{reg}}}
\left[
\psi(\phi(m)) - r(m)\,\phi(m)
\right].
\]

Finally, rewriting the second term as $\mathbb E_{Q_{\mathrm{reg}}}[r(m)\,\phi(m)]=\mathbb E_Q[\phi(m)]$, we obtain
\[
\inf_{\delta} \mathcal L(\delta)
=
\inf_{\phi}
\left\{
\mathbb E_{Q_{\mathrm{reg}}}[\psi(\phi(m))]
-
\mathbb E_Q[\phi(m)]
\right\}.
\]

Rearranging gives the desired result.
\end{proof}

\begin{theorem}[Dual Representation of Optimal Risk]
\label{thm:dual-risk}
Under the above assumptions, the optimal risk satisfies
\begin{equation}
\inf_{\delta} \mathcal L(\delta)
=
\inf_{\phi} \mathcal L(\phi)
=
C
-
\sup_{\phi}
\left\{
\mathbb E_Q[\phi(m)]
-
\mathbb E_{Q_{\mathrm{reg}}}[\psi(\phi(m))]
\right\}.
\end{equation}
\end{theorem}

\begin{proof}
The first equality follows from Proposition~\ref{prop:score-reformulation-clean}. The second follows by rearranging \eqref{eq:risk-formal}.
\end{proof}

\begin{corollary}[Equivalence with Variational Regulator]
\label{cor:regulator-risk}
Let $\mathcal R(Q)$ be defined as in \eqref{eq:regulator-general}. Then
\begin{equation}
\mathcal R(Q)
=
\sup_{\phi}
\left\{
\mathbb E_Q[\phi(m)]
-
\mathbb E_{Q_{\mathrm{reg}}}[\psi(\phi(m))]
\right\}
=
C - \inf_{\delta} \mathcal L(\delta).
\end{equation}
\end{corollary}

\begin{proof}
The first equality is the definition of $\mathcal R(Q)$ in \eqref{eq:regulator-general}. Theorem~\ref{thm:dual-risk} gives $\inf_{\delta}\mathcal L(\delta)=C-\sup_{\phi}\{\mathbb E_Q[\phi(m)]-\mathbb E_{Q_{\mathrm{reg}}}[\psi(\phi(m))]\}$. Rearranging this identity yields the second equality.
\end{proof}

The function $\phi$ is a test statistic that scores messages according to their likelihood of being generated from $Q$ rather than $Q_{\mathrm{reg}}$. The function $\psi$ is determined by the loss and penalizes large values of $\phi$ under the reference distribution. The regulator $\mathcal R(Q)$ quantifies the optimal separability between the two distributions.

\subsubsection{Feature-Based Restriction}

The preceding variational regulator allows the discriminator to inspect the entire message $m$. In many moderation, censorship, deception-detection, and phishing-defense systems, however, the receiver observes only a feature representation. Let $\mathcal F:\mathcal X\to\mathcal Z$ be a measurable feature map, with $\mathcal Z\subseteq\mathbb R^d$. The feature map is the mathematical version of a detector pipeline: raw text is mapped to counts, scores, embeddings, classifier logits, URL-risk features, or other summaries. The feature-induced laws are the pushforwards
\begin{equation}
Q_{\mathcal F}(B):=Q(\mathcal F^{-1}(B)),
\qquad
Q_{\mathrm{reg},\mathcal F}(B):=Q_{\mathrm{reg}}(\mathcal F^{-1}(B)),
\end{equation}
for measurable $B\subseteq\mathcal Z$. A feature-restricted discriminator has the form $\phi(m)=h(\mathcal F(m))$, where $h:\mathcal Z\to\mathbb R$ is a score on feature space. The corresponding regulator is
\begin{equation}
\mathcal R_{\mathcal F}(Q)
=
\sup_{h}
\left\{
\mathbb E_Q[h(\mathcal F(m))]
-
\mathbb E_{Q_{\mathrm{reg}}}[\psi(h(\mathcal F(m)))]
\right\}.
\label{eq:feature-restricted-regulator}
\end{equation}

\begin{proposition}[Feature-Induced Divergence]
\label{prop:feature-induced-divergence}
Suppose the score class for $h$ is rich enough to approximate bounded measurable functions on $\mathcal Z$, and let $f(u):=\psi^*(u)-\psi^*(1)$. If $Q_{\mathcal F}\ll Q_{\mathrm{reg},\mathcal F}$, then
\begin{equation}
\mathcal R_{\mathcal F}(Q)
=
D_f(Q_{\mathcal F}\,\|\,Q_{\mathrm{reg},\mathcal F})
+
\psi^*(1)
=
\mathbb E_{Q_{\mathrm{reg},\mathcal F}}
\left[
f\!\left(
\frac{dQ_{\mathcal F}}{dQ_{\mathrm{reg},\mathcal F}}(z)
\right)
\right]
+
\psi^*(1).
\label{eq:feature-divergence}
\end{equation}
\end{proposition}

\begin{proof}
For any measurable $h$, the change-of-variables identity for pushforward measures gives $\mathbb E_Q[h(\mathcal F(m))]=\mathbb E_{Q_{\mathcal F}}[h(z)]$ and $\mathbb E_{Q_{\mathrm{reg}}}[\psi(h(\mathcal F(m)))]=\mathbb E_{Q_{\mathrm{reg},\mathcal F}}[\psi(h(z))]$. Therefore \eqref{eq:feature-restricted-regulator} is the same variational discriminator problem as \eqref{eq:regulator-general}, but on feature space. Applying Theorem~\ref{thm:regulator-divergence} on $\mathcal Z$ gives \eqref{eq:feature-divergence}.
\end{proof}

\begin{corollary}[Information Loss under Feature Restriction]
\label{cor:feature-data-processing}
Under the assumptions of Proposition~\ref{prop:feature-induced-divergence}, suppose also that $Q\ll Q_{\mathrm{reg}}$. Then the restricted regulator cannot exceed the full message-level regulator:
\begin{equation}
\mathcal R_{\mathcal F}(Q)
\le
\mathcal R(Q).
\end{equation}
If $f$ is strictly convex, equality holds precisely when the likelihood ratio $dQ/dQ_{\mathrm{reg}}$ is determined by the feature $\mathcal F(m)$, up to $Q_{\mathrm{reg}}$-null sets.
\end{corollary}

\begin{proof}
Let $r(m):=dQ/dQ_{\mathrm{reg}}(m)$ and $Z:=\mathcal F(M)$ for $M\sim Q_{\mathrm{reg}}$. The feature-level likelihood ratio satisfies
$dQ_{\mathcal F}/dQ_{\mathrm{reg},\mathcal F}(Z)=\mathbb E_{Q_{\mathrm{reg}}}[r(M)\mid Z]$. Hence Jensen's inequality gives
\begin{align}
\mathcal R_{\mathcal F}(Q)
&=
\psi^*(1)
+
\mathbb E_{Q_{\mathrm{reg}}}
\left[
f\!\left(\mathbb E[r(M)\mid Z]\right)
\right]
\notag\\
&\le
\psi^*(1)
+
\mathbb E_{Q_{\mathrm{reg}}}[f(r(M))].
\end{align}
The right-hand side is $\mathcal R(Q)$ by Theorem~\ref{thm:regulator-divergence}. For strictly convex $f$, equality in Jensen's inequality occurs exactly when $r(M)$ is almost surely constant conditional on $Z$, which means that the likelihood ratio is a function of $\mathcal F(m)$.
\end{proof}

This result makes the role of feature design explicit. A coarse feature map produces a regulator that sees only the discrepancy between $Q_{\mathcal F}$ and $Q_{\mathrm{reg},\mathcal F}$, while a sufficient feature map preserves all likelihood-ratio information relevant to discrimination. In plain terms, a detector cannot penalize what its features cannot see. The case studies below use exactly this reduction: complete messages are mapped into receiver-side semantic classes, and the equilibrium is evaluated through the induced class distribution and score statistic rather than through every token sequence.

\section{Zero-Sum Game and Its Saddle-Point}

\subsection{Generator--Regulator Game}

We model the interaction between the generator and the regulator as a zero-sum game between a distribution $Q \in \Delta(\mathcal X)$ and a regulator score $\phi \in \Phi$. This formulation uses the variational regulator $\mathcal R(Q)$ from \eqref{eq:regulator-general}: the generator seeks useful and diverse messages, while the regulator seeks evidence that $Q$ deviates from the reference behavior encoded by $Q_{\mathrm{reg}}$.

The generator seeks to maximize semantic utility and entropy while paying a regulatory penalty. Using \eqref{eq:regulator-general}, the regulated objective is
\begin{equation}
\frac{1}{T} \mathbb E_Q[U_L(m)]
+
H(Q)
-
\lambda \mathcal R(Q),
\end{equation}
where $\lambda \ge 0$ controls the strength of regulation.

The parameter $\lambda$ should be interpreted as a policy or enforcement weight. When $\lambda=0$, the generator behaves as if the regulator did not exist and chooses the entropy-regularized utility maximizer. As $\lambda$ increases, the generator must give up utility or diversity whenever those gains require moving too far from the regulated reference distribution.

Substituting \eqref{eq:regulator-general} into the regulated objective and using the identity $-\sup = \inf(-)$, we obtain the min--max formulation
\begin{equation}
\sup_{Q \in \Delta(\mathcal X)}
\inf_{\phi \in \Phi}
\left\{
\frac{1}{T} \mathbb E_Q[U_L(m)]
+
H(Q)
-
\lambda \mathbb E_Q[\phi(m)]
+
\lambda \mathbb E_{Q_{\mathrm{reg}}}[\psi(\phi(m))]
\right\}.
\label{eq:game-revised}
\end{equation}

This is a zero-sum game: the generator selects $Q$ to maximize utility and entropy, while the regulator selects $\phi$ to penalize deviations from the reference behavior encoded in $Q_{\mathrm{reg}}$. The game-theoretic intuition is that the generator and regulator have opposing objectives, but their choices are coupled. If the generator shifts probability mass toward high-utility but easily detectable messages, the regulator can increase the penalty by choosing a discriminator that recognizes those messages. If the regulator is strong, the generator must move closer to $Q_{\mathrm{reg}}$, even when doing so sacrifices utility. The saddle point describes a stable compromise, in the same minimax sense that underlies zero-sum game theory \cite{vonneumann1928theorie,sion1958general} and strategic security models \cite{zhu2018networksecuritytutorial,zhang2020securedsvm}: neither side can improve its objective by changing strategy alone.

\subsection{Saddle-Point Structure and Optimal Generator}

We formalize the saddle-point structure of the generator--regulator game in \eqref{eq:game-revised} and characterize the optimal generator distribution. The theorem below has three roles. First, it states conditions under which a stable generator--regulator compromise exists. Second, it shows that once the regulator score $\phi$ is fixed, the generator is again a Gibbs law, but with a modified utility $U_L/T-\lambda\phi$. Third, it reduces the search for equilibrium to an optimization over the regulator score.

\begin{theorem}[Saddle-Point and Optimal Generator]
\label{thm:saddle-generator}
Consider the min--max problem in \eqref{eq:game-revised}. Assume that $\Delta(\mathcal X)$ is compact, $\Phi$ is convex, and $\psi$ is convex and lower semicontinuous. Then:

\begin{enumerate}
\item[(i)] The problem admits a saddle point $(Q^*, \phi^*)$, and the order of optimization can be exchanged, so that $\sup_{Q \in \Delta(\mathcal X)} \inf_{\phi \in \Phi}(\cdot)=\inf_{\phi \in \Phi} \sup_{Q \in \Delta(\mathcal X)}(\cdot)$.

\item[(ii)] For any fixed $\phi \in \Phi$, the inner maximization over $Q$ admits a unique solution given by
\begin{equation}
Q^*(m \mid \phi)
=
\frac{
\exp\!\left(
\frac{1}{T} U_L(m) - \lambda \phi(m)
\right)
}{
Z(\phi)
},
\label{eq:gibbs-phi}
\end{equation}
where
$Z(\phi)=\sum_{m \in \mathcal X}\exp(U_L(m)/T-\lambda\phi(m))$.

\item[(iii)] Substituting \eqref{eq:gibbs-phi} into \eqref{eq:game-revised}, the problem reduces to the convex optimization
\begin{equation}
\inf_{\phi \in \Phi}
\left\{
\log Z(\phi)
+
\lambda \mathbb E_{Q_{\mathrm{reg}}}[\psi(\phi(m))]
\right\}.
\label{eq:dual-phi}
\end{equation}
\end{enumerate}
\end{theorem}

\begin{proof}
We proceed in three steps.

\emph{Step 1 (Saddle-point existence).}
The objective in \eqref{eq:game-revised} is concave in $Q$ because the entropy $H(Q)$ is strictly concave and the remaining terms are linear in $Q$. It is convex in $\phi$ since $\psi$ is convex and the dependence on $\phi$ is affine in the remaining terms. Under the stated compactness and convexity assumptions, the existence of a saddle point and the interchange of $\sup$ and $\inf$ follow from Sion's minimax theorem.

\emph{Step 2 (Optimization over $Q$).}
Fix $\phi \in \Phi$. The optimization over $Q$ becomes
\begin{equation}
\sup_{Q \in \Delta(\mathcal X)}
\left\{
\mathbb E_Q\!\left[
\frac{1}{T} U_L(m) - \lambda \phi(m)
\right]
+
H(Q)
\right\}.
\end{equation}
This is a strictly concave optimization problem. Introducing a Lagrange multiplier for the normalization constraint $\sum_m Q(m)=1$, the first-order optimality condition yields
$\log Q^*(m)=U_L(m)/T-\lambda\phi(m)-\log Z(\phi)$, which gives \eqref{eq:gibbs-phi}. Uniqueness follows from strict concavity of the entropy.

\emph{Step 3 (Reduced problem).}
Substituting $Q^*(\cdot \mid \phi)$ into the objective yields $\sup_Q(\cdot)=\log Z(\phi)$ by the standard variational characterization of the log-partition function. Adding the regulator's remaining reference term gives \eqref{eq:dual-phi}.
\end{proof}

\subsection{Special Case: KL Regulation}

We specialize the variational regulator to the KL case. This case is especially readable because it produces a closed-form equilibrium in which the regulated generator is a geometric interpolation between a utility-based Gibbs law and the regulated reference law.

\begin{proposition}[KL Regulator]
\label{prop:kl-regulator}
Let $\psi(t) = e^t$. Then the regulator functional reduces to the Kullback--Leibler divergence, $\mathcal R(Q)=D_{\mathrm{KL}}(Q \,\|\, Q_{\mathrm{reg}})$.
\end{proposition}

\begin{proof}
The convex conjugate of $\psi(t)=e^t$ is $\psi^*(u)=u\log u-u+1$. Substituting this choice into the divergence representation \eqref{eq:f-divergence} yields
\[
\mathcal R(Q)
=
\mathbb E_{Q_{\mathrm{reg}}}
\left[
\frac{dQ}{dQ_{\mathrm{reg}}}(m)
\log \frac{dQ}{dQ_{\mathrm{reg}}}(m)
\right]
=
D_{\mathrm{KL}}(Q \,\|\, Q_{\mathrm{reg}}),
\]
which proves the claim.
\end{proof}

\begin{proposition}[Generator Problem under KL Regulation]
\label{prop:kl-generator}
Under KL regulation, the generator solves
\begin{equation}
\sup_{Q \in \Delta(\mathcal X)}
\left\{
\frac{1}{T} \mathbb E_Q[U_L(m)]
+
H(Q)
-
\lambda D_{\mathrm{KL}}(Q \,\|\, Q_{\mathrm{reg}})
\right\}.
\end{equation}
\end{proposition}

\begin{proof}
Under the KL choice $\psi(t)=e^t$, Proposition~\ref{prop:kl-regulator} gives $\mathcal R(Q)=D_{\mathrm{KL}}(Q\,\|\,Q_{\mathrm{reg}})$. Substituting this identity into the regulated generator objective $\frac{1}{T}\mathbb E_Q[U_L(m)]+H(Q)-\lambda\mathcal R(Q)$ gives the displayed optimization problem.
\end{proof}

\begin{proposition}[Closed-Form Equilibrium Distribution]
\label{prop:kl-solution}
The optimization problem in Proposition~\ref{prop:kl-generator} admits a unique maximizer given by
\begin{equation}
Q^*(m)
=
\frac{
Q_{\mathrm{reg}}(m)^{\frac{\lambda}{1+\lambda}}
\exp\!\left(
\frac{1}{T(1+\lambda)} U_L(m)
\right)
}{
Z_{\mathrm{reg}}
},
\end{equation}
where the normalization constant is
\begin{equation}
Z_{\mathrm{reg}}
=
\sum_{m \in \mathcal X}
Q_{\mathrm{reg}}(m)^{\frac{\lambda}{1+\lambda}}
\exp\!\left(
\frac{1}{T(1+\lambda)} U_L(m)
\right).
\end{equation}
\end{proposition}

\begin{proof}
Using the identity $D_{\mathrm{KL}}(Q \,\|\, Q_{\mathrm{reg}})=\mathbb E_Q[\log(Q(m)/Q_{\mathrm{reg}}(m))]$, the objective can be written as
\[
\mathbb E_Q\!\left[
\frac{1}{T} U_L(m)
-
\log Q(m)
-
\lambda \log \frac{Q(m)}{Q_{\mathrm{reg}}(m)}
\right].
\]

Rearranging terms yields
\[
\mathbb E_Q\!\left[
\frac{1}{T} U_L(m)
-
(1+\lambda)\log Q(m)
+
\lambda \log Q_{\mathrm{reg}}(m)
\right].
\]

This is a strictly concave optimization problem over $Q$. Introducing a Lagrange multiplier for the normalization constraint and taking first-order conditions, we obtain $\log Q^*(m)=U_L(m)/(T(1+\lambda))+(\lambda/(1+\lambda))\log Q_{\mathrm{reg}}(m)-\log Z_{\mathrm{reg}}$, which yields the stated form.
\end{proof}

The formula in Proposition~\ref{prop:kl-solution} is the clearest expression of the paper's main mechanism. The factor $\exp(U_L(m)/(T(1+\lambda)))$ rewards useful messages, while the factor $Q_{\mathrm{reg}}(m)^{\lambda/(1+\lambda)}$ pulls probability mass toward regulated behavior. Thus $\lambda$ does not simply delete undesirable messages; it continuously reweights the whole distribution over messages.

\begin{proposition}[Structure of the Equilibrium]
\label{prop:kl-structure}
The optimal distribution satisfies
\begin{equation}
\log Q^*(m)
=
\frac{1}{T(1+\lambda)} U_L(m)
+
\frac{\lambda}{1+\lambda} \log Q_{\mathrm{reg}}(m)
+ \mathrm{const}.
\end{equation}
\end{proposition}

\begin{proof}
Taking logarithms in the closed-form expression of Proposition~\ref{prop:kl-solution} gives $\log Q^*(m)=U_L(m)/(T(1+\lambda))+(\lambda/(1+\lambda))\log Q_{\mathrm{reg}}(m)-\log Z_{\mathrm{reg}}$. Since $Z_{\mathrm{reg}}$ does not depend on $m$, the term $-\log Z_{\mathrm{reg}}$ is the additive constant in the statement.
\end{proof}

\begin{corollary}[Limiting Regimes of the KL-Regulated Equilibrium]
\label{cor:kl-limits}
The equilibrium distribution in Proposition~\ref{prop:kl-solution} is a tilted Gibbs measure. As $\lambda \to 0$, it reduces to the standard Gibbs distribution proportional to $\exp(U_L(m)/T)$; as $\lambda \to \infty$, it converges to $Q_{\mathrm{reg}}$, assuming $Q_{\mathrm{reg}}(m)>0$ on the message space.
\end{corollary}

\begin{proof}
In the formula of Proposition~\ref{prop:kl-solution}, the exponent on $Q_{\mathrm{reg}}(m)$ is $\lambda/(1+\lambda)$ and the utility coefficient is $1/(T(1+\lambda))$. The first coefficient tends to $0$ and the second to $1/T$ as $\lambda\to0$; the first tends to $1$ and the second to $0$ as $\lambda\to\infty$. Normalization gives the two stated limits.
\end{proof}

Thus $\lambda$ controls the tradeoff between utility maximization and regulatory alignment: the utility term promotes high-value messages, while the reference-law term incorporates regulated behavior as a prior.

\subsection{Regulated Free Energy Density and Central Limit Theorem}

The KL-regulated equilibrium is still a distribution over an exponentially large message space. To study its large-$L$ behavior, we summarize the normalization scale by a free energy density. This quantity records the per-token growth rate of the regulated partition function and lets us compare long messages without tracking the full combinatorial sum. We define the regulated free energy density as
\begin{equation}
\mathcal F_{\mathrm{reg}}(T)
=
\lim_{L \to \infty}
\frac{1}{L}
\log
\sum_{m \in \mathcal V^L}
Q_{\mathrm{reg}}(m)^{\frac{\lambda}{1+\lambda}}
\exp\!\left(
\frac{1}{T(1+\lambda)} U_L(m)
\right),
\label{eq:free-energy-reg}
\end{equation}
provided that the limit exists.

The equilibrium distribution \(Q^*\) satisfies
\begin{equation}
\frac{1}{L} \log Q^*(m)
=
\frac{1}{T(1+\lambda)} \bar U_L(m)
+
\frac{\lambda}{1+\lambda} \frac{1}{L}\log Q_{\mathrm{reg}}(m)
-
\mathcal F_{\mathrm{reg}}(T)
+ o(1),
\label{eq:log-density-reg}
\end{equation}
as \(L \to \infty\), where \(\bar U_L(m) := U_L(m)/L\).

Thus, \(\mathcal F_{\mathrm{reg}}(T)\) plays the role of a normalization constant (pressure), ensuring that \(Q^*\) is properly normalized at the exponential scale. In statistical-mechanics language, it is the pressure; in information-theoretic language, it is the normalized log-partition value; in the present paper, it is the per-token value of regulated generation.

\medskip

We establish a central limit theorem for the log-density under the optimal generator. The reason a CLT is natural here is that many receiver scores and log-densities are sums of token-level contributions. When no single token dominates and conditional variances stabilize, the average score of a long message fluctuates approximately like a Gaussian random variable around its equilibrium mean.

\begin{theorem}[CLT for General Generator--Regulator Equilibrium]
\label{thm:clt-general}
Consider the equilibrium distribution $Q^*(m)\propto \exp(U_L(m)/T-\lambda\phi(m))$.

Assume:

\begin{enumerate}
\item[(G1)] (Additive structure) $U_L(m)=\sum_{k=1}^L u(w_k,h_k)$ and $\phi(m)=\sum_{k=1}^L \phi_k(w_k,h_k)$.

\item[(G2)] The process $\{(w_k,h_k)\}_{k\ge1}$ is adapted to a filtration $\{\mathcal F_k\}$.

\item[(G3)] (Lindeberg condition) For $Y_k:=u(w_k,h_k)/T-\lambda\phi_k(w_k,h_k)$,
\[
\frac{1}{L}\sum_{k=1}^L
\mathbb E\!\left[
\big(Y_k - \mathbb E[Y_k \mid \mathcal F_{k-1}]\big)^2
\mathbf{1}\{|Y_k|>\varepsilon \sqrt{L}\}
\right]
\to 0
\quad \forall \varepsilon>0.
\]

\item[(G4)] (Conditional variance convergence)
\[
\frac{1}{L}\sum_{k=1}^L
\mathrm{Var}(Y_k \mid \mathcal F_{k-1})
\;\xrightarrow{P}\;
\sigma^2 > 0.
\]
\end{enumerate}

Then, as \(L \to \infty\),
\begin{equation}
\sqrt{L}
\left(
\frac{1}{L}\log Q^*(m) - \mu
\right)
\;\xrightarrow{d}\;
\mathcal N(0,\sigma^2),
\end{equation}
where $\mu=\mathbb E_{Q^*}[Y_1]-\mathcal F_{\mathrm{reg}}(T)$.
\end{theorem}

\begin{proof}
Under (G1), the log-density admits the decomposition
\[
\log Q^*(m)
=
\sum_{k=1}^L Y_k
-
\log Z_L,
\]
where
\[
Z_L = \sum_{m \in \mathcal V^L}
\exp\!\left(
\frac{1}{T}U_L(m) - \lambda \phi(m)
\right).
\]
Dividing by \(L\), we obtain
\[
\frac{1}{L}\log Q^*(m)
=
\frac{1}{L}\sum_{k=1}^L Y_k
-
\frac{1}{L}\log Z_L.
\]

By definition of the regulated free energy density, $(1/L)\log Z_L \to \mathcal F_{\mathrm{reg}}(T)$. Define the martingale difference sequence $D_k:=Y_k-\mathbb E[Y_k\mid\mathcal F_{k-1}]$.
Under (G3)--(G4), the martingale central limit theorem yields
\[
\frac{1}{\sqrt{L}}\sum_{k=1}^L D_k
\;\xrightarrow{d}\;
\mathcal N(0,\sigma^2).
\]

Moreover,
\[
\frac{1}{L}\sum_{k=1}^L Y_k
=
\frac{1}{L}\sum_{k=1}^L \mathbb E[Y_k \mid \mathcal F_{k-1}]
+
\frac{1}{L}\sum_{k=1}^L D_k,
\]
and the predictable term converges to \(\mathbb E_{Q^*}[Y_1]\).

Combining the above yields the result.
\end{proof}

\section{Keyword-Based Safety Regulation}

We now specialize the feature-restricted regulator in \eqref{eq:feature-restricted-regulator} to a token-level score. This specialization connects the variational formulation to binary classification and statistical detection. Let $\mathcal K \subseteq \mathcal V$ be a finite set of regulated tokens, and define a token-level scoring function $s:\mathcal V\to\mathbb R$ by $s(w):=\kappa\,\mathbf 1\{w\in\mathcal K\}$, where $\kappa\ge0$. For a message $m=(w_1,\dots,w_L)\in\mathcal X$, define the regulated-token fraction $f_{\mathcal K}(m):=L^{-1}\sum_{k=1}^L\mathbf 1\{w_k\in\mathcal K\}$. The aggregate score is $S(m):=L^{-1}\sum_{k=1}^L s(w_k)=\kappa f_{\mathcal K}(m)$.

The quantity $S(m)$ summarizes the presence of regulated tokens and serves as a sufficient statistic for classification. This is the simplest version of the feature-based regulator: instead of analyzing all words and all semantics, the receiver compresses a message to one number. More realistic filters use many features, but the one-dimensional case is pedagogically useful because it makes the saddle point, threshold rule, and CLT visible in closed form.

\subsection{Binary Classification Game}

We formulate regulation as a binary classification game between a generator and a classifier. This is the score-restricted analogue of the decision problem in \eqref{eq:risk-formal}, with the classifier allowed to use only the statistic $S(m)$.

\begin{definition}[Binary Classification Model]
\label{def:binary-model}
Let $\mathcal X$ denote the message space and let $Q_{\mathrm{reg}} \in \Delta(\mathcal X)$ be a fixed reference distribution. The generator selects a distribution $Q \in \Delta(\mathcal X)$.

We consider the hypothesis testing problem $H_0:m\sim Q_{\mathrm{reg}}$ versus $H_1:m\sim Q$.
\end{definition}

\begin{definition}[Classifier and Restricted Class]
\label{def:classifier}
A classifier is a measurable function $\delta:\mathcal X\to[0,1]$, where $\delta(m)$ denotes the probability of assigning $m$ to class $H_1$.

Let $S(m)$ denote the aggregate score. We restrict attention to classifiers of the form $\delta(m)=g(S(m))$ for some measurable function $g:\mathbb R\to[0,1]$. The admissible class is
\[
\mathcal D_{\mathcal K}
=
\left\{
\delta : \delta(m)=g(S(m))
\right\}.
\]
\end{definition}

\begin{definition}[Classification Risk]
\label{def:risk}
Let $\ell_0, \ell_1 : [0,1] \to \mathbb R$ be convex loss functions. The classification risk is defined as
\begin{equation}
\mathcal L(Q,\delta)
=
\mathbb E_{Q_{\mathrm{reg}}}[\ell_0(\delta(m))]
+
\mathbb E_{Q}[\ell_1(\delta(m))].
\end{equation}
\end{definition}

\begin{definition}[Binary Classification Game]
\label{def:game}
The interaction between the generator and the classifier is defined as the zero-sum game
\begin{equation}
\sup_{Q \in \Delta(\mathcal X)}
\inf_{\delta \in \mathcal D_{\mathcal K}}
\mathcal J(Q,\delta),
\label{eq:binary-game}
\end{equation}
where the payoff functional is
\begin{equation}
\mathcal J(Q,\delta)
:=
\frac{1}{T}\mathbb E_Q[U_L(m)]
+
H(Q)
-
\lambda \mathcal L(Q,\delta),
\end{equation}
with $\lambda \ge 0$.
\end{definition}

 \begin{theorem}[Saddle-Point Characterization]
\label{thm:saddle-characterization}
Consider the binary classification game in Definition~\ref{def:game}. Assume:

\begin{enumerate}
\item the loss functions $\ell_0,\ell_1$ are convex, differentiable, and strictly convex,
\item $U_L$ is bounded,
\item $\mathcal D_{\mathcal K}$ is convex and closed,
\item $\Delta(\mathcal X)$ is compact.
\end{enumerate}

Then there exists a saddle point $(Q^*,\delta^*)$ satisfying the following conditions:

\emph{(i) Optimal classifier.}
For fixed $Q^*$, the optimal classifier $\delta^*$ solves
\begin{equation}
\delta^*(m)
\in
\argmin_{\delta \in \mathcal D_{\mathcal K}}
\left\{
\mathbb E_{Q_{\mathrm{reg}}}[\ell_0(\delta(m))]
+
\mathbb E_{Q^*}[\ell_1(\delta(m))]
\right\}.
\end{equation}

Moreover, for each $m$, $\delta^*(m)$ satisfies the pointwise first-order condition
\begin{equation}
\ell_0'(\delta^*(m)) \, dQ_{\mathrm{reg}}(m)
+
\ell_1'(\delta^*(m)) \, dQ^*(m)
= 0,
\label{eq:delta-foc}
\end{equation}
whenever the densities exist.

\emph{(ii) Optimal generator.}
For fixed $\delta^*$, the optimal generator $Q^*$ is given by
\begin{equation}
Q^*(m)
=
\frac{
\exp\!\left(
\frac{1}{T}U_L(m) - \lambda \ell_1(\delta^*(m))
\right)
}{
Z(\delta^*)
},
\end{equation}
where
\begin{equation}
Z(\delta^*)
=
\sum_{m \in \mathcal X}
\exp\!\left(
\frac{1}{T}U_L(m) - \lambda \ell_1(\delta^*(m))
\right).
\end{equation}

\emph{(iii) Saddle-point property.}
The pair $(Q^*,\delta^*)$ satisfies
\begin{equation}
\mathcal J(Q,\delta^*)
\le
\mathcal J(Q^*,\delta^*)
\le
\mathcal J(Q^*,\delta),
\quad
\forall Q,\delta.
\end{equation}

\end{theorem}

\begin{proof}
The existence of a saddle point follows from the same minimax argument used in Theorem~\ref{thm:saddle-generator}. 

(i) For fixed $Q^*$, the classifier problem is convex and separable across $m$, yielding the pointwise first-order condition \eqref{eq:delta-foc}.

(ii) For fixed $\delta^*$, the generator maximizes a strictly concave functional in $Q$. Introducing a Lagrange multiplier for normalization yields
\[
\log Q^*(m)
=
\frac{1}{T}U_L(m)
-
\lambda \ell_1(\delta^*(m))
-
\log Z(\delta^*),
\]
which gives the stated Gibbs form.

(iii) Follows from the saddle-point optimality conditions.
\end{proof}

\begin{corollary}[Monotone and Threshold Structure of the Optimal Classifier]
\label{cor:threshold}
Consider the binary classification game in Theorem~\ref{thm:saddle-characterization}. Suppose that:

\begin{enumerate}
\item the classifier is restricted to $\mathcal D_{\mathcal K} = \{\delta(m)=g(S(m))\}$,
\item the loss functions $\ell_0,\ell_1$ are convex and strictly proper,
\item the likelihood ratio $\frac{dQ^*}{dQ_{\mathrm{reg}}}(m)$ is a monotone function of the score $S(m)$ (monotone likelihood ratio condition).
\end{enumerate}

Then the optimal classifier $\delta^*$ is monotone in the score $S(m)$. In particular, there exists a non-decreasing function $g^* : \mathbb R \to [0,1]$ such that
$\delta^*(m)=g^*(S(m))$.

Moreover, in the case of deterministic classification, such as $0$--$1$ loss, the optimal classifier admits a threshold structure:
$\delta^*(m)=\mathbf{1}\{S(m)\ge\eta^*\}$ for some threshold $\eta^* \in \mathbb R$.
\end{corollary}

\begin{proof}
For fixed generator law $Q^*$, Theorem~\ref{thm:saddle-characterization} reduces the classifier's problem to pointwise Bayes decisions under $Q_{\mathrm{reg}}$ and $Q^*$. Strict propriety implies that the optimal probabilistic decision is a monotone transform of the posterior probability that the message was drawn from $Q^*$ rather than from $Q_{\mathrm{reg}}$. This posterior is itself a monotone transform of the likelihood ratio $dQ^*/dQ_{\mathrm{reg}}(m)$. By the monotone likelihood-ratio assumption, the likelihood ratio is monotone in $S(m)$, so the optimal classifier can be written as $\delta^*(m)=g^*(S(m))$ for a non-decreasing function $g^*$. In the deterministic $0$--$1$ case, Bayes optimality compares the same likelihood ratio with a constant threshold determined by class priors and misclassification costs. Monotonicity in $S(m)$ turns that likelihood-ratio threshold into a score threshold, so $\delta^*(m)=\mathbf{1}\{S(m)\ge\eta^*\}$ for some $\eta^*$.
\end{proof}

\begin{proposition}[Token-Level Exponential Tilting under Additive Receiver Loss]
\label{prop:token-tilting}
Consider the saddle-point characterization in Theorem~\ref{thm:saddle-characterization}. Assume:

\begin{enumerate}
\item the utility is additive, $U_L(m)=\sum_{k=1}^L u(w_k,h_k)$,
\item the induced receiver loss is additive at equilibrium, $\ell_1(\delta^*(m))=\alpha\sum_{k=1}^L s(w_k)+\beta$ with $\alpha>0$,
\item the statistic $s(w_k)$ is the token-level receiver score.
\end{enumerate}

Then the optimal generator $Q^*$ admits an additive (energy-based) representation:
\begin{equation}
Q^*(m)
\propto
\exp\!\left(
\sum_{k=1}^L
\Big[
\frac{1}{T}u(w_k,h_k)
-
\theta\, s(w_k)
\Big]
\right),
\label{eq:token-tilted}
\end{equation}
for $\theta=\lambda\alpha>0$.

Consequently, for a fixed context $h$, the induced token distribution is an exponentially tilted version of the unregulated token distribution: $q^*(w\mid h)\propto q_0(w\mid h)\exp(-\theta s(w))$, where $q_0(w\mid h)\propto \exp(u(w,h)/T)$.

\end{proposition}

\begin{proof}
The generator component of Theorem~\ref{thm:saddle-characterization} gives $Q^*(m)\propto \exp(U_L(m)/T-\lambda\ell_1(\delta^*(m)))$. Substituting the additive utility and additive receiver loss gives
\begin{equation*}
Q^*(m)
\propto
\exp\!\left(
\sum_{k=1}^L
\left[
\frac{u(w_k,h_k)}{T}
-
\lambda\alpha s(w_k)
\right]
-
\lambda\beta
\right).
\end{equation*}
The constant $-\lambda\beta$ is absorbed into the normalizing factor, yielding \eqref{eq:token-tilted} with $\theta=\lambda\alpha$. Conditioning on a fixed context $h$ and comparing the tilted one-token factor with the unregulated factor $q_0(w\mid h)\propto\exp(u(w,h)/T)$ gives the final proportionality.
\end{proof}

The corollary shows that optimal detection depends only on the scalar statistic $S(m)$ via a monotone threshold rule. The proposition shows that, under an additive receiver loss, the generator responds by exponentially downweighting tokens with high score $s(w)$. Together, these results characterize the equilibrium as a one-dimensional detection problem coupled with an exponential tilting of the token distribution.

\begin{theorem}[CLTs for Score, Log-Density, and Surprisal under Lindeberg Condition]
\label{thm:clt-unified}
Let $Q^*$ be the optimal generator. Define $\xi_k:=u(w_k,h_k)/T-\theta s(w_k)$ and $X_k:=-\log Q^*(w_k \mid h_k)$. Let $\mathcal F_{\mathcal K}(T)$ denote the keyword-regulated free energy density, defined by the normalization identity $\log Z_L^{\mathcal K}(T)=L\,\mathcal F_{\mathcal K}(T)+o(L)$, where $Z_L^{\mathcal K}(T)$ is the partition function associated with the tilted representation in \eqref{eq:token-tilted}.

Assume:

\begin{enumerate}
\item[(A1)] The process $\{(w_k,h_k)\}_{k\ge 1}$ is adapted to a filtration $\{\mathcal F_k\}$.

\item[(A2)] (Lindeberg condition) For each sequence $Y_k \in \{s(w_k), \xi_k, X_k\}$,
\[
\frac{1}{L}\sum_{k=1}^L 
\mathbb E\!\left[
\big(Y_k - \mathbb E[Y_k \mid \mathcal F_{k-1}]\big)^2
\mathbf{1}\{|Y_k|>\varepsilon \sqrt{L}\}
\right]
\to 0
\quad \forall \varepsilon>0.
\]

\item[(A3)] (Conditional variance convergence) For each sequence $Y_k$,
\[
\frac{1}{L}\sum_{k=1}^L 
\mathrm{Var}(Y_k \mid \mathcal F_{k-1})
\;\xrightarrow{P}\;
\sigma_Y^2.
\]

\item[(A4)] The log-density admits the decomposition $\log Q^*(m)=\sum_{k=1}^L \xi_k - L\,\mathcal F_{\mathcal K}(T)$.
\end{enumerate}

Then, as $L \to \infty$, the following central limit theorems hold:

\emph{(i) Score.}
Let $S_L:=L^{-1}\sum_{k=1}^L s(w_k)$ and $\mu_S:=\mathbb E_{Q^*}[s(w_1)]$.
Then
\begin{equation}
\sqrt{L}\big(S_L - \mu_S\big)
\;\xrightarrow{d}\;
\mathcal N(0,\sigma_S^2),
\end{equation}
where $\sigma_S^2$ is the limit in (A3) with $Y_k = s(w_k)$.

\emph{(ii) Log-density.}
Let $\Lambda_L:=L^{-1}\log Q^*(m)$ and $\mu_\Lambda:=\mathbb E_{Q^*}[\xi_1]-\mathcal F_{\mathcal K}(T)$.
Then
\begin{equation}
\sqrt{L}\big(\Lambda_L - \mu_\Lambda\big)
\;\xrightarrow{d}\;
\mathcal N(0,\sigma_\Lambda^2),
\end{equation}
where $\sigma_\Lambda^2$ corresponds to $Y_k = \xi_k$.

\emph{(iii) Surprisal.}
Let $\bar X_L:=L^{-1}\sum_{k=1}^L X_k$ and $H(Q^*):=\mathbb E_{Q^*}[X_1]$.
Then
\begin{equation}
\sqrt{L}\big(\bar X_L - H(Q^*)\big)
\;\xrightarrow{d}\;
\mathcal N(0,\sigma_X^2),
\end{equation}
where $\sigma_X^2$ corresponds to $Y_k = X_k$.

\end{theorem}

\begin{proof}
For each choice of $Y_k \in \{s(w_k), \xi_k, X_k\}$, define the martingale difference sequence $D_k:=Y_k-\mathbb E[Y_k\mid\mathcal F_{k-1}]$.
Under assumptions (A2)--(A3), the martingale central limit theorem applies, yielding convergence of
\[
\frac{1}{\sqrt{L}}\sum_{k=1}^L D_k
\;\xrightarrow{d}\;
\mathcal N(0,\sigma_Y^2).
\]
Since
\[
\frac{1}{L}\sum_{k=1}^L Y_k
=
\frac{1}{L}\sum_{k=1}^L \mathbb E[Y_k \mid \mathcal F_{k-1}]
+
\frac{1}{L}\sum_{k=1}^L D_k,
\]
and the predictable term converges to $\mu_Y$, the result follows.

The log-density case uses (A4).
\end{proof}

The three limits correspond respectively to detection statistics, log-density (energy), and information content, showing that all key observables of the optimal generator exhibit Gaussian fluctuations under the same structural conditions. For interpretation, the score CLT approximates filtering probability, the log-density CLT approximates how likely a generated message is under the equilibrium law, and the surprisal CLT approximates how concentrated the information content is around its entropy rate.

\begin{corollary}[Gaussian Approximation of Detection]
\label{cor:clt-threshold}
Under the conditions of Theorem~\ref{thm:clt-unified}, we have
\begin{equation}
\sqrt{L}\big(S_L - \mu_S\big)
\;\xrightarrow{d}\;
\mathcal N(0,\sigma_S^2).
\end{equation}
Equivalently, $S_L=\mu_S+(\sigma_S/\sqrt L)Z_L$, where $Z_L \xrightarrow{d} \mathcal N(0,1)$. 
In particular, for any fixed threshold $\eta^* \in \mathbb R$,
\begin{equation}
\mathbb P_{Q^*}\big(S_L \ge \eta^*\big)
=
\mathbb P\!\left(
Z_L \ge \frac{\eta^* - \mu_S}{\sigma_S/\sqrt{L}}
\right),
\end{equation}
Therefore, for large $L$,
\begin{equation}
\mathbb P_{Q^*}\big(S_L \ge \eta^*\big)
\approx
1-\Phi\!\left(
\frac{(\eta^* - \mu_S)\sqrt{L}}{\sigma_S}
\right)
\quad \text{(Gaussian approximation)}.
\end{equation}
\end{corollary}

\begin{proof}
The first statement follows directly from Theorem~\ref{thm:clt-unified}. The representation with $Z_L$ follows by rescaling. The probability statement is obtained by rewriting the event in terms of $Z_L$ and applying the Gaussian approximation to the resulting tail probability.
\end{proof}

\begin{corollary}[Concentration of Average Surprisal]
\label{cor:surprisal-concentration}
Under the conditions of Theorem~\ref{thm:clt-unified}, $\bar X_L=H(Q^*)+O_P(L^{-1/2})$.
\end{corollary}

\begin{proof}
From Theorem~\ref{thm:clt-unified}, $\sqrt{L}(\bar X_L-H(Q^*))\xrightarrow{d}\mathcal N(0,\sigma_X^2)$, which implies tightness of $\sqrt{L}(\bar X_L - H(Q^*))$. Hence $\bar X_L-H(Q^*)=O_P(L^{-1/2})$.
\end{proof}

\begin{corollary}[Gaussian Approximation of Surprisal]
\label{cor:surprisal-gaussian}
Under the conditions of Theorem~\ref{thm:clt-unified},
\begin{equation}
\sqrt{L}\big(\bar X_L - H(Q^*)\big)
\;\xrightarrow{d}\;
\mathcal N(0,\sigma_X^2),
\end{equation}
and hence
\begin{equation}
\bar X_L
=
H(Q^*) + \frac{\sigma_X}{\sqrt{L}}Z_L,
\qquad
Z_L \xrightarrow{d} \mathcal N(0,1).
\end{equation}
\end{corollary}

\begin{proof}
This is an immediate restatement of Theorem~\ref{thm:clt-unified} for $Y_k = X_k$, together with a rescaling argument.
\end{proof}

The first corollary shows that, under the optimal generator, the high-dimensional detection problem reduces to a one-dimensional Gaussian test based on the statistic $S_L$. The classifier operates as a threshold test on an approximately normal variable, with mean $\mu_S$ controlled by the generator through exponential tilting. This provides a second-order finite-length characterization of detectability.

The probability $\mathbb P_{Q^*}(S_L \ge \eta^*)$ quantifies the likelihood of detection. The generator influences this probability through $\mu_S$ and $\sigma_S^2$, which depend on the tilting parameter $\theta$. Increasing regulation (larger $\lambda$) shifts $\mu_S$ downward, reducing detection probability at the cost of utility. The CLT reveals an explicit tradeoff between utility and detectability at the $\sqrt{L}$ scale.

The surprisal $\bar X_L$ concentrates around the entropy rate $H(Q^*)$ with fluctuations of order $L^{-1/2}$. This shows that, under the optimal generator, the information content per token becomes sharply predictable for large $L$. In particular, typical sequences satisfy
$-L^{-1}\log Q^*(m)\approx H(Q^*)$, linking the CLT to the asymptotic equipartition property.

The central limit theorems for $S_L$, $\Lambda_L$, and $\bar X_L$ show that detection statistics, log-density, and information content all exhibit Gaussian fluctuations around their respective means. These quantities are coupled through the same underlying token-level structure:
$\xi_k=u(w_k,h_k)/T-\theta s(w_k)$ and $X_k=-\log Q^*(w_k\mid h_k)$. Thus, the equilibrium simultaneously controls detection (via $S_L$), typicality (via $\Lambda_L$), and information (via $\bar X_L$).

The Gaussian approximation is accurate for moderate deviations of order $O(L^{-1/2})$. For rare events such as $\mathbb P(S_L \ge \eta)$ with $\eta$ far from $\mu_S$, large deviations theory is required, leading to exponential decay of the form $\exp(-L I(\eta))$. Thus, the CLT provides a local approximation, while large deviations characterize tail behavior.

\begin{corollary}[Berry--Esseen Bound for Surprisal]
\label{cor:be-surprisal}
If, in addition, the normalized third moments are uniformly bounded and the dependence conditions are strong enough for a Berry--Esseen estimate, then with $Y_k = X_k$,
\begin{equation}
\sup_{x}
\left|
\mathbb P\!\left(
\sqrt{L}\frac{\bar X_L - H(Q^*)}{\sigma_X} \le x
\right)
-
\Phi(x)
\right|
\le
\frac{C}{\sqrt{L}}.
\end{equation}
\end{corollary}

\begin{proof}
Apply the Berry--Esseen estimate to the normalized partial sum with $Y_k=X_k$. The bounded normalized third moments and the stated dependence assumptions ensure that the accumulated third absolute moment is of order $L$, while the variance of $\sum_{k=1}^L X_k$ is $L\sigma_X^2$ under the normalization in Theorem~\ref{thm:clt-unified}. The resulting Kolmogorov distance between $\sqrt{L}(\bar X_L-H(Q^*))/\sigma_X$ and the standard normal distribution is therefore bounded by $C/\sqrt L$ for a constant $C$ independent of $L$.
\end{proof}

\section{Censorship Filtering Case Study}
\label{sec:censorship-case-study}

The first case study models a receiver-side censorship filter that blocks, suppresses, or down-ranks messages containing restricted topics, names, slogans, URLs, images, or event-specific terms. This setting is empirically grounded in measurement studies of online information control \cite{king2013censorship,roberts2018censored} and in public blocklist datasets collected by Citizen Lab. The Citizen Lab \texttt{chat-censorship} repository contains keyword blocklists and other trigger material used in applications and platforms studied across chat, live streaming, mobile games, open-source projects, WeChat, QQMail, search, translation, and related services \cite{citizenlabChatCensorship}. Several of these measurements distinguish between lists recovered by reverse engineering and lists discovered by sample testing, which is directly aligned with the present model: the receiver implements a scoring rule over observable message features, while the generator changes the distribution of message classes in response.

The goal of the case study is not to reproduce any particular platform's censorship rule. Instead, it gives a calibrated finite-vocabulary abstraction in which the semantic classes correspond to receiver-side features that can be extracted from a real dataset. The class $c_0^{(C)}$ represents ordinary permissible language, $c_1^{(C)}$ represents indirect or context-bearing language that is semantically related to a restricted topic but does not directly match a known trigger, and $c_2^{(C)}$ represents direct restricted terms, names, slogans, URLs, image hashes, or event labels. A practical extractor may combine blocklist matching, named-entity recognition, contextual topic classification, and embedding similarity to restricted-topic exemplars. If $z_j(m)$ denotes the resulting feature indicators, then a simple receiver score has the form $S(m)=\sum_j\theta_j z_j(m)$; a learned moderation system can replace this linear score with a classifier while preserving the same distributional interpretation.

\begin{table}[t]
\centering
\begingroup
\singlespacing
\scriptsize
\begin{tabular}{p{0.28\textwidth}|p{0.58\textwidth}}
Receiver-side feature & Censorship-filtering interpretation \\
\hline
Recovered blocklist or exact keyword match & A high score is assigned when a message contains a term, name, slogan, URL, or other trigger appearing in a recovered or experimentally discovered list. This is the direct restricted class $c_2^{(C)}$. \\
Tested trigger family or event-specific list membership & A message receives high score when it matches terms associated with a time-varying event, policy campaign, public figure, platform-specific rule, or experimentally measured censorship category. \\
Semantic proximity to restricted-topic exemplars & A moderate score is assigned when a message is close to restricted-topic examples under an embedding or topic model even if it avoids exact lexical matches. This is the indirect class $c_1^{(C)}$. \\
Contextual coordination or dissemination signal & Additional score can be assigned when otherwise ordinary language appears in a context associated with collective action, repeated posting, event organization, or coordinated circulation.
\end{tabular}
\endgroup
\caption{Receiver-side feature map for the censorship-filtering case study. The table describes how empirical censorship measurements can be converted into abstract class labels and scores; it does not enumerate restricted terms or provide evasion templates.}
\label{tab:censorship-filtering-features}
\end{table}

We use a three-class KL-regulated equilibrium. For the censorship application, the class set is $\{c_0^{(C)},c_1^{(C)},c_2^{(C)}\}$, the receiver-accepted reference law is $q_{\mathrm{reg}}^{(C)}$, the class utility vector is $u^{(C)}$, and the receiver risk-score vector is $r^{(C)}$. With $T=1$, Proposition~\ref{prop:kl-solution} gives the class-level equilibrium
\begin{equation}
q_\lambda^{(a)}(c)
=
\frac{
q_{\mathrm{reg}}^{(a)}(c)^{\lambda/(1+\lambda)}
\exp\!\left(u^{(a)}(c)/(1+\lambda)\right)
}{
\sum_{c'}
q_{\mathrm{reg}}^{(a)}(c')^{\lambda/(1+\lambda)}
\exp\!\left(u^{(a)}(c')/(1+\lambda)\right)
}.
\label{eq:case-class-distribution}
\end{equation}
where $a=C$ in the present section and $a=P$ in Section~\ref{sec:phishing-case-study}. For a length-$L$ message with class sequence $(C_1,\dots,C_L)$, define the receiver score and flag probability by
\begin{equation}
S_L^{(a)}
=
\frac{1}{L}\sum_{k=1}^L r^{(a)}(C_k),
\qquad
P_{\mathrm{flag}}^{(a)}(\lambda)
=
\mathbb P_{q_\lambda^{(a)}}\!\left(S_L^{(a)}\ge\eta_a\right).
\label{eq:case-flag-probability}
\end{equation}
Because the benchmark is independent at the class level, this probability can be computed exactly from the multinomial distribution of the class counts. For longer or dependent messages, the CLT in Theorem~\ref{thm:clt-unified} gives the approximation $P_{\mathrm{flag}}^{(a)}(\lambda)\approx 1-\Phi((\eta_a-\mu_S^{(a)})\sqrt{L}/\sigma_S^{(a)})$, where $\mu_S^{(a)}=\mathbb E_{q_\lambda^{(a)}}[r^{(a)}(C)]$ and $(\sigma_S^{(a)})^2=\mathrm{Var}_{q_\lambda^{(a)}}(r^{(a)}(C))$. The plots combine analytic and sampled quantities. Solid curves are computed from \eqref{eq:case-class-distribution} and the exact multinomial law of the receiver score. Shaded bands show finite-sample variability, not uncertainty about the equilibrium: for each $\lambda$, the script draws $400$ batches, each with $25$ messages of length $L=100$ sampled from $q_\lambda^{(a)}$, and plots one empirical standard deviation of the batch estimates. This mirrors a small-language-model experiment in which continuations are sampled with the Transformers library \cite{wolf2020transformers}, mapped to receiver-side classes, and summarized by empirical class counts rather than operationally sensitive message text.

The finite-class benchmark also allows a direct verification of the score CLT. For each application, define the standardized score $Z_L^{(a)}=\sqrt{L}(S_L^{(a)}-\mu_S^{(a)})/\sigma_S^{(a)}$. Since the class counts are multinomial, the exact distribution of $Z_L^{(a)}$ and the exact flag probability are computable without asymptotic approximation. Figure~\ref{fig:case-study-clt-verification} compares these exact quantities with the Gaussian approximation from Theorem~\ref{thm:clt-unified}. The agreement is strongest in the transition regime where the threshold is near the mean score, while discrepancies in very small tail probabilities reflect the usual distinction between central-limit and large-deviation approximations.

\begin{figure}[t]
\centering
\includegraphics[width=\textwidth]{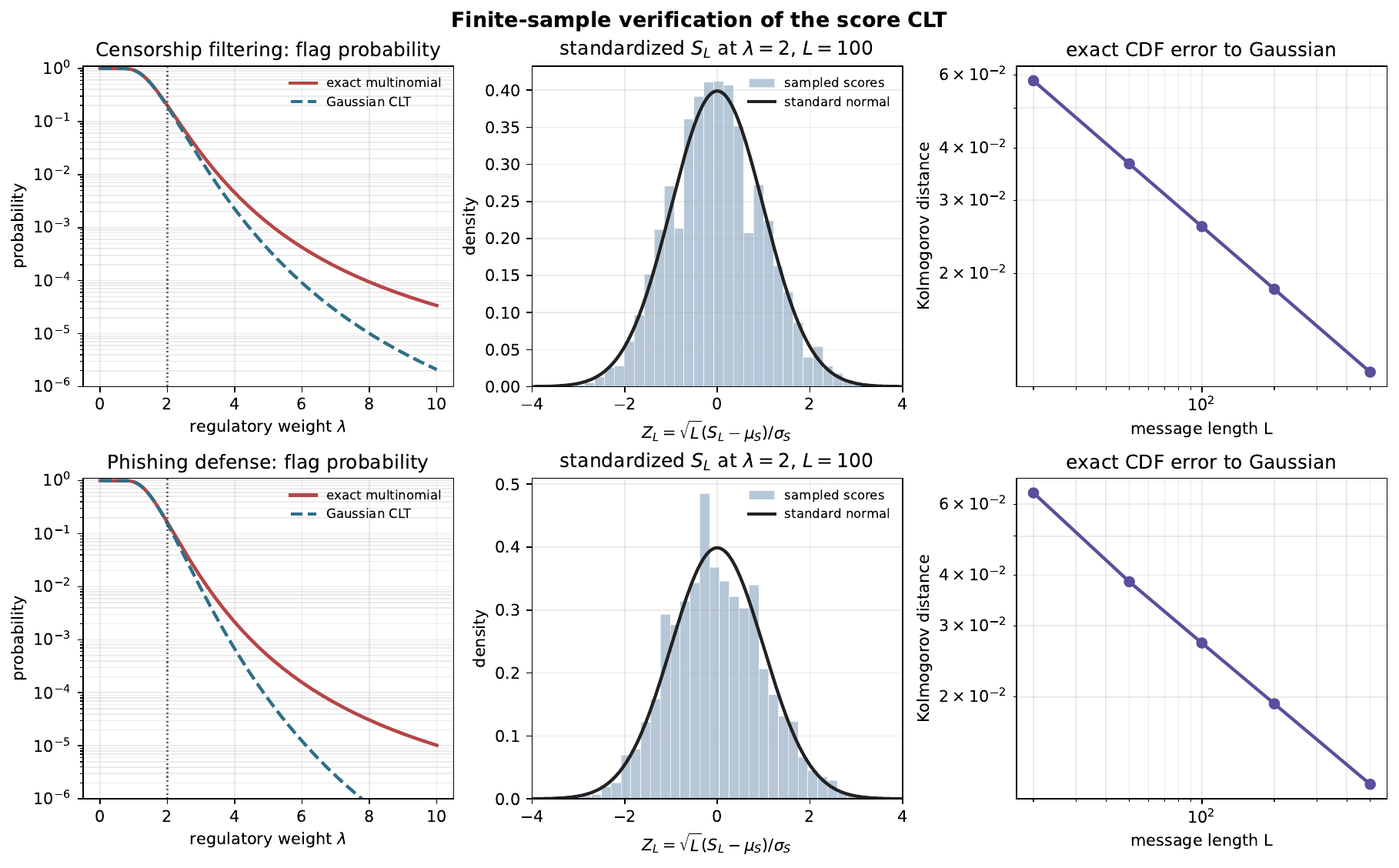}
\caption{Finite-sample verification of the score CLT in the censorship and phishing case studies. The left panels compare exact multinomial flag probabilities with the Gaussian CLT approximation for $L=100$. The middle panels show histograms of $20{,}000$ sampled standardized scores at $\lambda=2$ and $L=100$, with the standard normal density overlaid. The right panels plot the exact Kolmogorov distance between the distribution of $Z_L^{(a)}$ and the standard normal law at $\lambda=2$ for $L\in\{20,50,100,200,500\}$.}
\label{fig:case-study-clt-verification}
\end{figure}

After optimizing out the KL regulator, the generator chooses a class distribution $q\in\Delta(\{c_0^{(a)},c_1^{(a)},c_2^{(a)}\})$ to maximize the reduced objective
\begin{equation}
J_\lambda^{(a)}(q)
=
\langle q,u^{(a)}\rangle
+H(q)
-\lambda D_{\mathrm{KL}}\!\left(q\,\|\,q_{\mathrm{reg}}^{(a)}\right).
\label{eq:case-reduced-objective}
\end{equation}
The maximizer of \eqref{eq:case-reduced-objective} is the equilibrium distribution in \eqref{eq:case-class-distribution}. Equivalently, the saddle point is the point on the probability simplex where utility, entropy, and regulatory divergence jointly balance.

\begin{figure}[t]
\centering
\includegraphics[width=\textwidth]{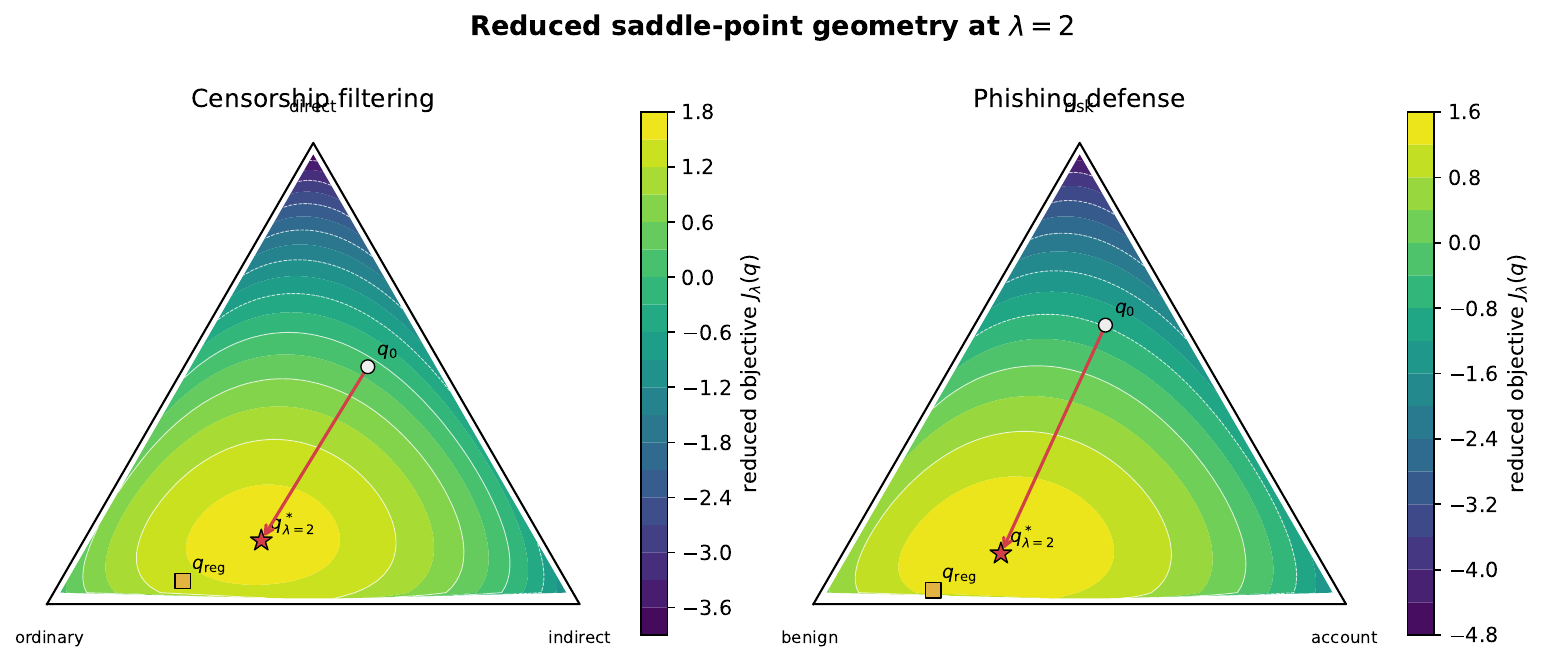}
\caption{Reduced saddle-point geometry for the censorship and phishing case studies at $\lambda=2$. Each triangle is the probability simplex over the three semantic classes. Color contours show the reduced objective $J_\lambda^{(a)}(q)$ after the KL regulator has been optimized. The circle marks the unregulated Gibbs distribution $q_0$, the square marks the receiver-accepted reference law $q_{\mathrm{reg}}$, and the star marks the regulated equilibrium $q_{\lambda=2}^*$. The arrow shows how regulatory pressure moves the generator away from the utility-dominated distribution and toward the receiver-compatible region of the simplex.}
\label{fig:case-study-saddlepoint-simplex}
\end{figure}

The generator's utility is communicative: ordinary words receive low reward, topic-bearing language receives higher reward, and direct restricted terms receive the highest reward. The receiver assigns moderate score to indirect expressions and high score to direct restricted expressions, capturing the incentive to reshape wording away from direct terms. The reference law is $q_{\mathrm{reg}}^{(C)}=(0.72,0.23,0.05)$, the utility vector is $u^{(C)}=(0.4,1.3,1.7)$, and the receiver score is $r^{(C)}=(0,0.25,1)$. We take $L=100$ and $\eta_C=0.25$, so a message is flagged when its average score is at least $0.25$. These numbers are stylized but empirically interpretable: $q_{\mathrm{reg}}^{(C)}$ represents the accepted mixture of ordinary, indirect, and direct-topic language under a censoring platform; $u^{(C)}$ represents the communicative value of topic-bearing expression; and $r^{(C)}$ represents the receiver's risk score derived from the features in Table~\ref{tab:censorship-filtering-features}.

For a sanitized instance, suppose a receiver maps a length-$20$ message to counts $(12,6,2)$, where $n_j$ is the number of positions assigned to $c_j^{(C)}$. The score is $S^{(C)}=(0\cdot 12+0.25\cdot 6+1\cdot 2)/20=0.175$, below $\eta_C=0.25$, so the message is not filtered. Counts $(10,5,5)$ instead give $S^{(C)}=0.3125$, so the message is filtered. Filtering is therefore a function of receiver-side feature counts, not a judgment about one word in isolation.

At the regulated equilibrium $\lambda=2$, Table~\ref{tab:case-study-censorship} gives $q_2^{(C)}\approx(0.529,0.334,0.138)$. A readable length-$20$ sample from this law can be represented by the sanitized message
\begin{quote}
\emph{Local residents plan to meet after work to discuss service delays, share a memorial notice, mention [RESTRICTED-NAME], reference [RESTRICTED-EVENT], and preserve archived notes labeled [RESTRICTED-TOPIC] and [RESTRICTED-SLOGAN].}
\end{quote}
The bracketed terms are placeholders for receiver-side feature extraction. If the filter maps the message to counts $(n_0,n_1,n_2)=(10,6,4)$, the score is $S^{(C)}=(0\cdot 10+0.25\cdot 6+1\cdot 4)/20=0.275$, so this equilibrium sample is filtered because $S^{(C)}>\eta_C$. Equilibrium does not mean every sample avoids filtering; it specifies the distribution under which filtered and unfiltered realizations occur.

\begin{table}[t]
\centering
\begingroup
\singlespacing
\scriptsize
\begin{tabular}{c|ccc|cccc}
$\lambda$ &
$q(c_0^{(C)})$ &
$q(c_1^{(C)})$ &
$q(c_2^{(C)})$ &
$\bar u$ &
$H$ &
$D_{\mathrm{KL}}$ &
$P_{\mathrm{flag}}$ \\
\hline
$0$   & $0.140$ & $0.345$ & $0.515$ & $1.380$ & $0.985$ & $1.111$ & $1.000$ \\
$0.5$ & $0.310$ & $0.386$ & $0.303$ & $1.142$ & $1.092$ & $0.486$ & $0.9999$ \\
$1$   & $0.418$ & $0.371$ & $0.211$ & $1.008$ & $1.061$ & $0.254$ & $0.9309$ \\
$2$   & $0.529$ & $0.334$ & $0.138$ & $0.879$ & $0.976$ & $0.100$ & $0.2004$ \\
$5$   & $0.632$ & $0.283$ & $0.085$ & $0.766$ & $0.857$ & $0.022$ & $1.20\times 10^{-3}$ \\
$10$  & $0.674$ & $0.259$ & $0.067$ & $0.720$ & $0.797$ & $0.006$ & $3.39\times 10^{-5}$
\end{tabular}
\endgroup
\caption{Censorship-filtering case study. The generator values direct restricted expressions most, but the receiver-side censor assigns those expressions the largest score. Increasing $\lambda$ shifts mass toward the receiver-accepted reference distribution, lowers the filter probability, and reduces communicative utility.}
\label{tab:case-study-censorship}
\end{table}

Table~\ref{tab:case-study-censorship} shows how regulation changes the generator distribution. At $\lambda=0$, the generator ignores the reference law and assigns more than half of its mass to direct restricted terms, producing near-certain filtering. As $\lambda$ increases, mass moves toward ordinary and indirect classes. The intermediate class $c_1^{(C)}$ preserves some communicative utility while carrying less receiver-side risk, giving a stylized account of linguistic reshaping under censorship without enumerating evasion phrases.

\begin{figure}[t]
\centering
\includegraphics[width=\textwidth]{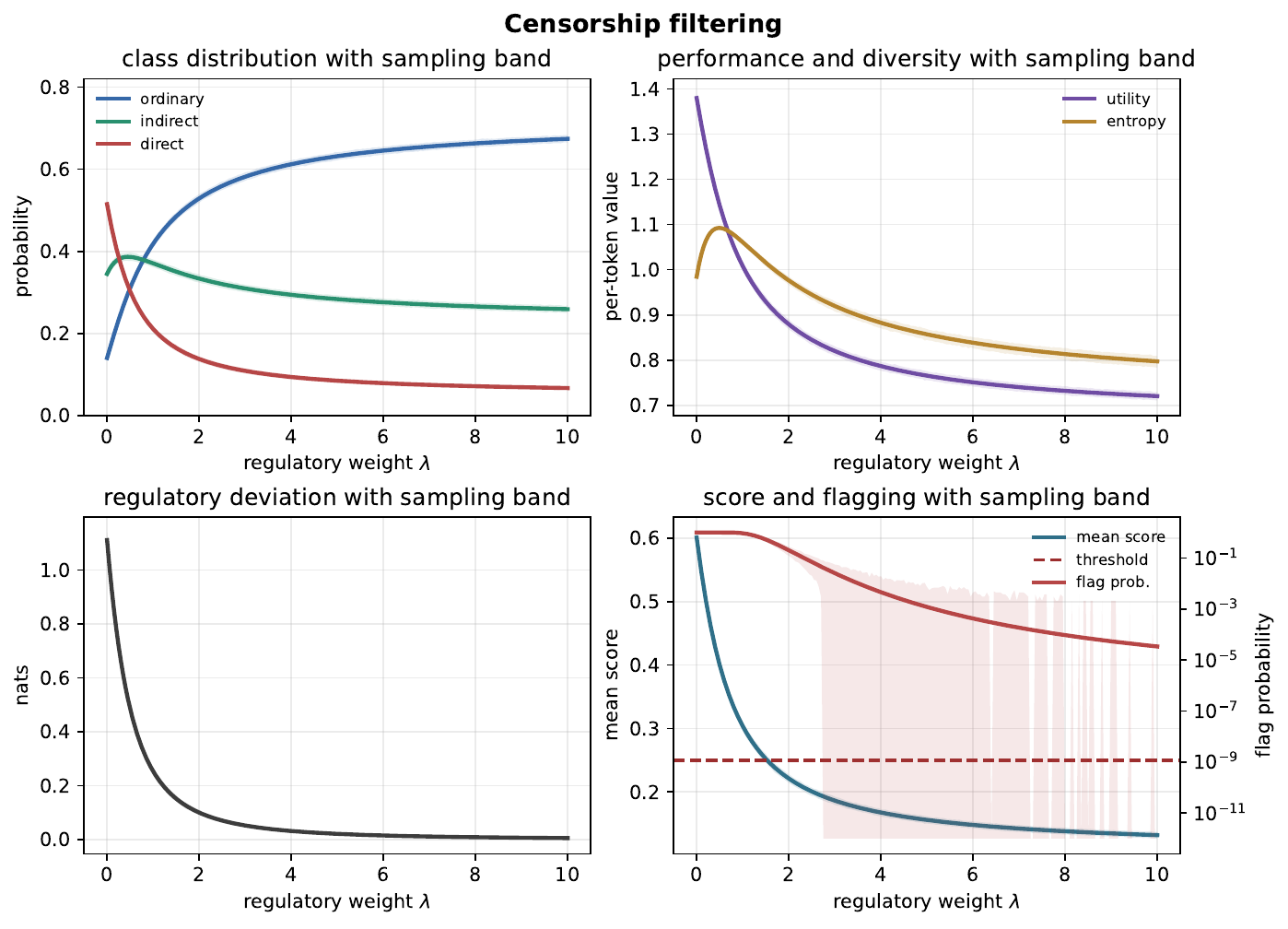}
\caption{Censorship-filtering tradeoffs as the regulatory weight $\lambda$ varies from $0$ to $10$. Solid curves report the analytic equilibrium and exact multinomial flag probability. Shaded bands show one empirical standard deviation over $400$ repeated sampling batches, each containing $25$ sampled messages of length $L=100$.}
\label{fig:case-study-censorship-paths}
\end{figure}

Figure~\ref{fig:case-study-censorship-paths} shows the full path. Regulatory pressure continuously reweights all classes, not only the direct restricted class. Entropy can initially rise as mass spreads across classes, then fall as $q_\lambda^{(C)}$ approaches $q_{\mathrm{reg}}^{(C)}$. Once the mean receiver score falls sufficiently below threshold, concentration over $100$ positions makes filtering rare; the shaded bands show the sampling variability visible near this transition in a finite empirical evaluation.

\section{Phishing-Defense Case Study}
\label{sec:phishing-case-study}

The phishing case study has the same mathematical structure but the opposite ethical interpretation: the generator is adversarial and the receiver is the defender. The receiver may be a human user supported by a warning system, an email-security gateway, or a browser-side classifier. Its goal is to filter social-engineering and credential-seeking risk, a setting studied in usable security, phishing detection, and deceptive-information resilience \cite{dhamija2006phishing,khonji2013phishing,yang2024deceptiveits}. Public repositories indexed under GitHub's \texttt{phishing-attacks} topic illustrate the breadth of the operational ecosystem: the topic aggregates repositories involving phishing kits, phishing pages, phishing domains, domain permutations, threat-intelligence feeds, indicators of compromise, phishing detection, and training or simulation tooling \cite{githubPhishingAttacksTopic}. We use this ecosystem only to motivate defensive feature categories. The case study does not execute attack code, name target organizations, specify domains, or provide instructions for phishing.

The class $c_0^{(P)}$ represents benign service communication, $c_1^{(P)}$ represents legitimate account-security language that may resemble warning or recovery messages, and $c_2^{(P)}$ represents credential-seeking or social-engineering risk. A realistic receiver-side extractor may combine natural-language cues with URL and infrastructure features, sender-authentication signals, landing-page similarity scores, attachment or form indicators, and known threat-intelligence signals. The key modeling point is that the detector scores a complete message through observable features, while the adversarial generator changes the class distribution in response to that detector.

\begin{table}[t]
\centering
\begingroup
\singlespacing
\scriptsize
\begin{tabular}{p{0.28\textwidth}|p{0.58\textwidth}}
Receiver-side feature & Phishing-defense interpretation \\
\hline
Credential or code request & A high score is assigned when the message asks for authentication material, account recovery secrets, payment credentials, or private verification information. This is the phishing-risk class $c_2^{(P)}$. \\
Suspicious link, domain, form, or attachment signal & A message receives score when it contains a risky destination, domain mismatch, shortened or newly observed domain, form-collection cue, or attachment pattern associated with compromise attempts. \\
Impersonation, authority, urgency, or threat cue & A moderate or high score is assigned when the message combines time pressure, institutional impersonation, account-lock language, or authority framing with an action request. \\
Benign security-notification pattern & A lower score is assigned when the message resembles legitimate account-security communication, such as an informational alert or routine notification, without credential collection or suspicious routing.
\end{tabular}
\endgroup
\caption{Receiver-side feature map for the phishing-defense case study. The entries describe defensive scoring features that can be estimated from threat-intelligence repositories, phishing corpora, or benign account-message data; they are not generation templates.}
\label{tab:phishing-filtering-features}
\end{table}

The reference distribution for receiver-accepted messages is $q_{\mathrm{reg}}^{(P)}=(0.76,0.21,0.03)$, the utility vector is $u^{(P)}=(0.5,1.0,1.9)$, and the receiver score is $r^{(P)}=(0,0.20,1)$. We again take $L=100$, and we set $\eta_P=0.20$. Here $q_{\mathrm{reg}}^{(P)}$ represents a benign mixture of service notifications and legitimate security messages, while $u^{(P)}$ represents the adversarial utility of persuasive, risk-bearing content. The score $r^{(P)}$ maps the features in Table~\ref{tab:phishing-filtering-features} into a receiver-side filter statistic.

For a defensive instance, counts $(n_0,n_1,n_2)=(13,5,2)$ give $S^{(P)}=(0\cdot 13+0.20\cdot 5+1\cdot 2)/20=0.15$, below $\eta_P=0.20$, corresponding to ordinary account-security language with few high-risk features. Counts $(11,5,4)$ give $S^{(P)}=0.25$, so the message is flagged. The same feature-count machinery now represents a benign receiver protecting the user.

At $\lambda=2$, Table~\ref{tab:case-study-phishing} gives $q_2^{(P)}\approx(0.593,0.297,0.110)$. A sanitized training-only realization from this law is
\begin{quote}
\emph{Training simulation: a sender claims an account review is pending, asks the recipient to use [EXTERNAL-LINK], says action is needed today, requests [AUTHENTICATION-CODE], and refers to [HELPDESK-IDENTITY].}
\end{quote}
This is not an operational phishing message: no organization, domain, form, or credential field is specified. It only displays features a receiver-side detector would score. If the detector assigns counts $(n_0,n_1,n_2)=(11,5,4)$, then $S^{(P)}=(0\cdot 11+0.20\cdot 5+1\cdot 4)/20=0.25>\eta_P$, so the message is flagged. An equilibrium realization with counts $(13,5,2)$ has score $0.15$ and is not flagged. Thus filtering remains stochastic at equilibrium: the saddle point determines feature-count probabilities, while the threshold determines which realizations are filtered.

\begin{table}[t]
\centering
\begingroup
\singlespacing
\scriptsize
\begin{tabular}{c|ccc|cccc}
$\lambda$ &
$q(c_0^{(P)})$ &
$q(c_1^{(P)})$ &
$q(c_2^{(P)})$ &
$\bar u$ &
$H$ &
$D_{\mathrm{KL}}$ &
$P_{\mathrm{flag}}$ \\
\hline
$0$   & $0.149$ & $0.246$ & $0.605$ & $1.470$ & $0.933$ & $1.613$ & $1.000$ \\
$0.5$ & $0.360$ & $0.328$ & $0.312$ & $1.101$ & $1.097$ & $0.608$ & $1.000$ \\
$1$   & $0.482$ & $0.325$ & $0.193$ & $0.933$ & $1.034$ & $0.282$ & $0.9467$ \\
$2$   & $0.593$ & $0.297$ & $0.110$ & $0.802$ & $0.913$ & $0.098$ & $0.1633$ \\
$5$   & $0.686$ & $0.255$ & $0.059$ & $0.710$ & $0.773$ & $0.019$ & $4.91\times 10^{-4}$ \\
$10$  & $0.722$ & $0.235$ & $0.043$ & $0.678$ & $0.712$ & $0.005$ & $1.02\times 10^{-5}$
\end{tabular}
\endgroup
\caption{Phishing-defense case study. The receiver is a defender: high scores correspond to credential-seeking or social-engineering risk. Increasing $\lambda$ moves the adversarial generator toward the receiver-accepted reference law and sharply reduces the probability that a length-$100$ message exceeds the defensive filter threshold.}
\label{tab:case-study-phishing}
\end{table}

Table~\ref{tab:case-study-phishing} illustrates the defender's role. With no regulatory pressure, the adversarial generator places most mass on the phishing-risk class and is almost always flagged. Larger $\lambda$ values move the generator closer to the benign reference law. The drop in $P_{\mathrm{flag}}$ is sharp because the threshold is applied to an average over $100$ positions; once the mean score falls below threshold, concentration makes crossings rare. In a small-language-model experiment, the table would be obtained by replacing hand-specified class probabilities with empirical counts from sampled outputs and a fixed defensive extractor.

\begin{figure}[t]
\centering
\includegraphics[width=\textwidth]{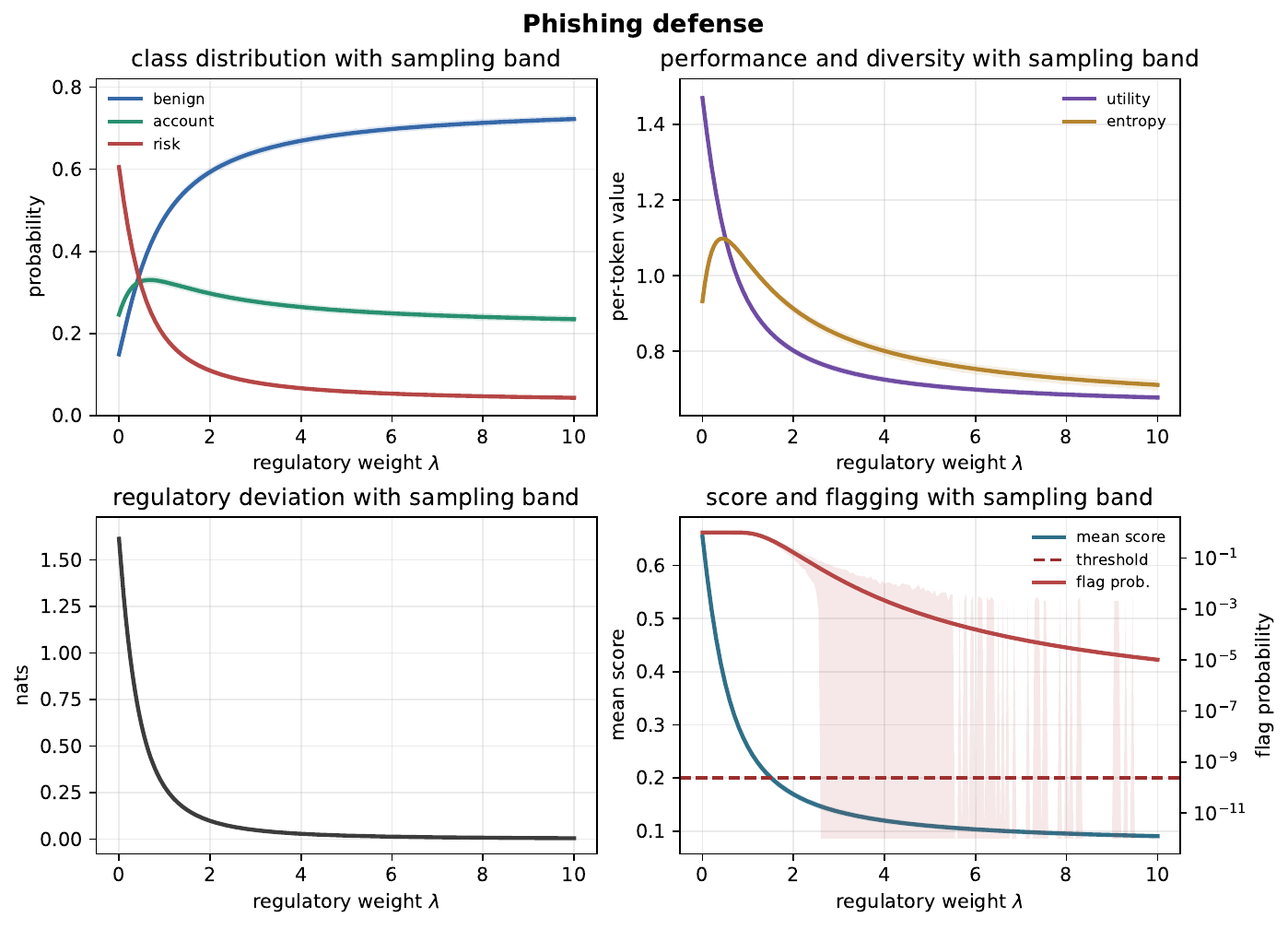}
\caption{Phishing-defense tradeoffs as the regulatory weight $\lambda$ varies from $0$ to $10$. Solid curves report the analytic equilibrium and exact multinomial flag probability. Shaded bands show one empirical standard deviation over $400$ repeated sampling batches, each containing $25$ sampled messages of length $L=100$.}
\label{fig:case-study-phishing-paths}
\end{figure}

Figure~\ref{fig:case-study-phishing-paths} shows the same formal path under the defensive interpretation. The unregulated generator favors the risk class because it has the largest adversarial utility. As $\lambda$ grows, probability shifts toward benign and legitimate account-security language. Near threshold, small distributional changes produce large changes in $P_{\mathrm{flag}}$; far below threshold, the multinomial tail is small even when the risk class has nonzero probability. The shaded bands show the empirical fluctuation from finite sampled messages.

\section{Conclusion}

This paper developed a variational framework for regulated language generation by representing the generator as a distribution $Q$ over complete messages. Autoregressive token sampling induces this distribution operationally, while entropy-regularized utility maximization gives an analytic Gibbs form. The regulator is a functional $\mathcal R(Q)$ that penalizes deviation from a regulated reference law $Q_{\mathrm{reg}}$; in variational form, it is an optimal discriminator whose value becomes an $f$-divergence by convex duality.

The generator--regulator interaction is a saddle-point problem. Under KL regulation, the equilibrium is the tilted Gibbs measure
\[
Q^*(m)
\propto
Q_{\mathrm{reg}}(m)^{\frac{\lambda}{1+\lambda}}
\exp\!\left(
\frac{1}{T(1+\lambda)}U_L(m)
\right).
\]
The parameter $\lambda$ controls the strength of the regulated reference, while $T$ controls utility-based concentration. The finite-length results show that scores, log-densities, and surprisals concentrate around deterministic means with fluctuations of order $L^{-1/2}$, so detection probabilities and typicality can often be approximated by one-dimensional Gaussian calculations.

The censorship and phishing case studies show how the formal objects become operational receiver-side filtering problems. By specifying a utility $U_L$, reference law $Q_{\mathrm{reg}}$, discriminator class, regulation strength $\lambda$, and score statistic $S_L$, the framework quantifies the tradeoff among generation quality, diversity, regulatory alignment, and detectability.

\bibliographystyle{unsrt}
\bibliography{ref}

\end{document}